\documentclass[11pt,oneside]{article}

\usepackage{amsthm, amsmath, amssymb, amsfonts, graphicx, epsfig, epstopdf}
\usepackage{enumerate}
\usepackage[colorlinks=true, pdfstartview=FitV, linkcolor=blue,
            citecolor=blue, urlcolor=blue]{hyperref}
\usepackage{url}
\makeatletter\def\url@leostyle{%
 \@ifundefined{selectfont}{\def\UrlFont{\sf}}{\def\UrlFont{\scriptsize\ttfamily}}} \makeatother

\newtheorem{theorem}{Theorem}[section]
\newtheorem{proposition}[theorem]{Proposition}
\newtheorem{lemma}[theorem]{Lemma}
\newtheorem{corollary}[theorem]{Corollary}

\theoremstyle{definition}
\newtheorem{definition}[theorem]{Definition}

\theoremstyle{remark}
\newtheorem{remark}[theorem]{Remark}

\def\bF{\mathbb{F}}
\def\bG{\mathbb{G}}

\def\bI{\mathbb{I}}

\newtheorem{assumption}[theorem]{Assumption}

\newcommand{\wtd}{\widetilde \delta}

\makeindex
\begin{document}

\title{Counterparty Risk and the Impact of Collateralization\\
in CDS Contracts}

\author{Tomasz R. Bielecki\footnote{TRB and IC acknowledge support from the NSF grant DMS-0908099}\\
\small{Department of Applied Mathematics,}\\[-0.1ex]
\small{Illinois Institute of Technology,}\\[-0.1ex]
\small{Chicago, 60616 IL, USA}\\[-0.1ex]
\url{bielecki@iit.edu}\\
\and
Igor Cialenco\footnotemark[\value{footnote}]\\[-0.1ex]
\small{Department of Applied Mathematics,}\\[-0.1ex]
\small{Illinois Institute of Technology,}\\[-0.1ex]
\small{Chicago, 60616 IL, USA}\\[-0.3ex]
\url{igor@math.iit.edu} \\
\and
Ismail Iyigunler \\
\small{Department of Applied Mathematics,}\\[-0.1ex]
\small{Illinois Institute of Technology,}\\[-0.1ex]
\small{Chicago, 60616 IL, USA}\\[-0.1ex]
\url{iiyigunl@iit.edu} \\[5ex]}

\date{First Circulated: April 11, 2011\\[1ex]
This version: \today}
\bigskip
\maketitle

\begin{abstract}
We analyze the counterparty risk embedded in CDS contracts, in presence of a bilateral margin agreement. First, we investigate the pricing of collateralized counterparty risk and we derive the bilateral Credit Valuation Adjustment (CVA), unilateral Credit Valuation Adjustment (UCVA) and Debt Valuation Adjustment (DVA). We propose a model for the collateral by incorporating all related factors such as the thresholds, haircuts and margin period of risk. We derive the dynamics of the bilateral CVA in a general form with related jump martingales. We also introduce the Spread Value Adjustment (SVA) indicating the counterparty risk adjusted spread. Counterparty risky and the counterparty risk-free spread dynamics are derived and the dynamics of the SVA is found as a consequence. We finally employ a Markovian copula model for default intensities and illustrate our findings with numerical results.
\end{abstract}
\newpage
\tableofcontents
\newpage
\section{Introduction}

Not very long after the collapse of prestigious institutions like Long-Term Capital Management, Enron and Global Crossing, the financial industry has again witnessed dramatic downfalls of financial institutions such as Lehman Brothers, Bear Stearns and Wachovia. These recent collapses have stressed out the importance of measuring, managing and mitigating counterparty risk appropriately.

Counterparty risk is defined as the risk that a party in an over-the-counter (OTC) contract will default and will not be able to honor its contractual obligations. Since the exchange-traded derivative contracts are subject to clearing by the exchange, counterparty risk arises from OTC derivatives only. The main challenge in the counterparty risk assessment and hedging is that the exposures of OTC derivatives are stochastic and involve dependencies and systemic risk factors such as wrong way risks; the additional level of complexity is introduced by risk mitigation techniques such as collateralization and netting. Therefore, one needs to model potential future exposures and to price the counterparty risk appropriately according to margin agreements that underlie the collateralization procedures.

Brigo and Capponi \cite{BrCapp} focuses on a Gaussian copula model and study bilateral counterparty risk using a CIR++ intensity model. Recently, Brigo et al. \cite{BCPP2011} extended this framework to the collateralized contracts with an application to interest rate swaps under bilateral margin agreements. Hull and White \cite{HullWhite} propose a static copula model and study unilateral counterparty risk on credit default swaps. Bielecki et al. \cite{Bielecki2011} study unilateral counterparty risk with the absence of any margin agreements. Assefa et al. \cite{Assefa2011} consider the portfolio of credit default swaps under Markovian copula model and consider only fully collateralized contracts. Jarrow and Yu \cite{JarrowYu} deal with the counterparty risk by using a dependence structure based on the default intensities of the counterparties. This approach, that also addresses the contagion risk issue, is considered in Leung and Kwok \cite{LeungKwok2005}. All these works mentioned above employ the reduced form modeling technology. However, structural models have also been used to model counterparty risk. Good examples of this approach are papers by Lipton and Sepp \cite{LiptonSepp} and Blanchet-Scalliet and Patras \cite{BCPatras}. Moreover, Stein and Lee \cite{SteinLee} study and illustrate credit valuation adjustment computations in the fixed income markets.

Various issues regarding the simulation of credit valuation adjustments under margin agreements are studied by Pykhtin in \cite{Pykhtin2009}. Furthermore, Cesari et al. \cite{Cesari2010} and Gregory \cite{Gregory2009} provide thorough treatments of the methods and the applications used in practice regarding the counterparty risk.

In this paper, we analyze the counterparty risk in a Credit Default Swap (CDS) contract in presence of a bilateral margin agreement. There are three risky names associated with the contract: the reference entity, protection seller (the counterparty) and the protection buyer (the investor). Contrary to the common approach which starts with defining the Potential Future Exposure (PFE) and derives the Credit Valuation Adjustment (CVA) as the price of the counterparty risk, we find the CVA as the difference between the market values of a counterparty risk-free and a counterparty risky CDS contracts and deduct the relevant credit exposures accordingly. We consider the problem of bilateral counterparty risk assessment; that is, we consider the situation where the two counterparties  of the CDS contract, i.e. the investor and the counterparty, are subject to default risk in a counterparty risky CDS contract.

We focus on the collateralized contracts, where there is a bilateral margin agreement is in force as a vital risk mitigation tool, and it requires the counterparty and the investor to post collateral in case their exposure exceeds specific threshold values. We propose a model for the collateral by incorporating all related factors, such as thresholds, margin period of risk and minimum transfer amount. Then, we derive the dynamics of the bilateral CVA which is essential for dynamic hedge of the counterparty risk. We also compute the decomposition of the fair spread for the CDS, and we analyze so called Spread Value Adjustment (SVA). Essentially, SVA represents the adjustment to be made to the fair spread to incorporate the counterparty risk into the CDS contract.

Using the bilateral CVA formula, we derive relevant formulas for assessment of credit exposures, such as PFE, Expected Positive Exposure (EPE) and Expected Negative Exposure (ENE).

In our model, the dependence between defaults and the wrong way risk is represented in a Markovian copula framework that accounts for simultaneous defaults among the three names represented in a CDS contract. In this way, our model takes broader systemic risk factors into account and quantifies the wrong way risk and the double defaults in a tangible manner.

This paper is organized as follows. In Section \ref{sec:CPrisk}, we first define the dividend processes regarding the counterparty risky and the counterparty risk-free CDS contract in case of a bilateral margin agreement. We also define the CVA, UCVA and the DVA terms as well as the credit exposures such as PFE, EPE and ENE. We then prove the dynamics of the CVA in Section \ref{sec:CVAdynamics}. Moreover, we find the fair spread adjustment term and its dynamics in Section \ref{sec:SVA}. In Section \ref{sec:Numerics}, we simulate  the collateralized exposures and the CVA using our Markovian copula model of default dependence.

\section{Pricing Counterparty Risk: CVA, UCVA and DVA}\label{sec:CPrisk}

We consider a standard CDS contract, and we label by $1$ the reference name, by $2$ the counterparty, and by $3$ the investor. Each of the three names may default before the maturity of the CDS contract, and we denote by ${\tau }_{{1}},$ ${\tau }_{2}$ and ${\tau }_{3}$ their respective default times. These times are modeled as non-negative random variables given on a underlying probability space $\left( \Omega ,\mathcal{G},\mathbb{Q}\right)$. We let $T$ and $\kappa$ to denote the maturity and the spread of our CDS contract, respectively. We assume the recovery at default covenant; that is, we assume that recoveries are paid at times of default.

We introduce right-continuous processes $H_{t}^{i}$ by setting $H_{t}^{i}=${$ \mathbb{I}$}$_{\left\{ {\tau }_{i}\leq t\right\} }$ and we denote by $ \mathbb{H}^{i}$ the associated filtrations so that $\mathcal{H}_{t}^{i}=\sigma \left( H_{u}^{i}:u\leq t\right) $ for $i=1,2,3.$

We assume that we are given a market filtration $\mathbb{F}$, and we define the enlarged filtration
$\mathbb{G}=\mathbb{F}{\vee }\mathbb{H}^{1}{\vee }\mathbb{H}^{2}{\vee }
\mathbb{H}^{3}$, that is $\mathcal{G}_{t}=\sigma \left( \mathcal{F}_{t}\cup
\mathcal{H}_{t}^{1}\cup \mathcal{H}_{t}^{2}\cup \mathcal{H}_{t}^{3}\right)$ for any $t\in \mathbb{R}_{+}.$ For each $t\in \mathbb{R}_{+}$ total information available at time $t$ is captured by the $\sigma$-field $\mathcal{G}_{t}.$ In particular, processes $H^{i}$ are $\mathbb{G}$-adapted and the random times ${\tau }_{{i}}$ are $\mathbb{G}$-stopping times for $i=1,2,3.$

Next, we define the first default time as the minimum of $\tau_1,$ $\tau_2$ and $\tau_3$: $\tau=\tau_1\wedge\tau_2\wedge\tau_3$; the corresponding indicator process is $H_t=\bI_{\{ \tau \leq t\} }$. In addition, we define the first default time of the two counterparties: $\hat{\tau}=\tau_2\wedge\tau_3$, and the corresponding  indicator process $\hat{H_t}=\bI_{\{ \hat{\tau} \leq t\} }$.

We also denote by $B_t$ the savings account process, that is
\begin{equation*}
B_{{t}}=e^{\int_{0}^{t}r_{s}ds},
\end{equation*}%
where the $\mathbb{F}$-progressively measurable process $r$ models the short-term interest rate. We also postulate that $\mathbb{Q}$ represents a martingale measure associated with the choice of the savings account $B$ as a discount factor (or numeraire).

\subsection{Dividend Processes and Marking-to-Market}

In this paper, all cash flows and the prices are considered from the perspective of the investor.

We start by introducing the \textsl{counterparty-risk-free} dividend process $D$, which describes all cash flows associated with a counterparty-risk-free CDS contract;\footnote{We shall sometimes refer to such contract as to the clean contract.} that is, $D$  does not account for the counterparty risk.

\begin{definition}
The cumulative dividend process $D$ of a counterparty risk-free CDS contract maturing at time $T$ is given as,
\begin{equation}
D_{t}=\int_{\left] 0,t\right] }\delta _{u}^{1}dH_{u}^{1}-\kappa \int_{\left]
0,t\right] }\left( 1-H_{u}^{1}\right) du,  \label{d1}
\end{equation}
for every $t\in \left[ 0,T\right] $, where $\delta ^{1}:\left[ 0,T\right] \rightarrow \mathbb{R}$ is an $\mathbb{F}$-predictable processes.
\end{definition}

Process $\delta ^{1}$ represents the loss given default (LGD);
that is $\delta ^{1}=1-R^1_t,$ where $R^1$ is the fraction of the nominal that is recovered in case of the default of the reference name. We assume unit nominal, for simplicity.

The ex-dividend price processes of the counterparty risk-free CDS contract, say $S$, describes the current market value, or the Mark-to-Market (MtM) value of this contract.

\begin{definition}
The ex-dividend price process $S$ of a counterparty risk-free CDS contract maturing at time $T$ is given by,%
\begin{equation}
S_{{t}}=B_{{t}}{\mathbb{E}}_{\mathbb{Q}}\left( \left. \int_{\left] {t},T
\right] }B_{u}^{-1}dD_{u}\right\vert \mathcal{G}_{{t}}\right), \ t\in [0,T]. \label{pr1}
\end{equation}
\end{definition}
\begin{remark}
Accordingly, we define the cumulative (dividend) price process, say $\widehat{S}$, of a counterparty risk-free CDS contract as
\begin{equation*}
\widehat{S}_{t}=S_{{t}}+B_{{t}}\int_{%
\left] {0},t\right] }B_{u}^{-1}dD_{u}, \ t\in [0,T].
\end{equation*}
\end{remark}

Now, we are in  position to define the dividend process $D^{C}$ of a counterparty risky CDS contract, that is the CDS contract that accounts for the counterparty risk associated with the two counterparties of the contract.

\begin{definition}
The dividend process $D^{C}$ of a $T$-maturity counterparty-risky CDS contract is given as
\begin{align}
D_{t}^{C} &=\int_{\left] 0,t\right] }C_{u}dH_{u}+\int_{\left] 0,t\right] }\widetilde{\delta }_{u}^{1}\left( 1-H_{u-}\right) dH_{u}^{1}
+\int_{\left] 0,t\right] }\widetilde{\delta }_{u}^{2}\left( 1-H_{u-}\right) dH_{u}^{2}\label{d2}  \nonumber \\
& \ \ \ +\int_{\left] 0,t\right] }\widetilde{\delta }_{u}^{3}\left( 1-H_{u-}\right)dH_{u}^{3}+\int_{\left] 0,t\right] }\widetilde{\delta }_{u}^{4}\left(1-H_{u-}\right) d[ H^{2},H^{3}] _{u} \\
& \ \ \ +\int_{\left] 0,t\right] }\widetilde{\delta }_{u}^{5}\left( 1-H_{u-}\right) d[ \hat{H},H^{1}] _{u}-\kappa \int_{\left] 0,t\right] }\left( 1-H_{u}\right) du, \ \ \ t\in [0,T],
\notag
\end{align}
where $\widetilde{\delta }^{i}:\left[ 0,T\right] \rightarrow \mathbb{R}$ is an $\mathbb{F}$-predictable processes for $i=1,\ldots ,5$ and $C:\left[ 0,T\right] \rightarrow \mathbb{R}$ is an $\bF$-predictable process representing the collateral amount kept in the margin account.
\end{definition}

Margin account is a contractual tool that supplements  the CDS contract so to reduce potential losses that may be incurred by one of the counterparties in case of the default of the other counterparty, while the CDS contract is still alive. For the detailed description of the mechanics of the collateral formation in the margin account we refer to Section \ref{marginSec} (see also \cite{Bielecki2011a}).

In case of any credit event, associated with the collateralized CDS contract, the first cash flow that takes place is the ``transfer'' of the collateral amount; for example, in case  when the underlying entity defaults at time ${t=\tau =\tau }_{1}$,  (before any of the counterparties defaults) the collateral in the margin account is acquired by one of the counterparties (depending on the sign of $C_\tau$). Thus, consistently with the conventions of so called close-out cash flows (cf. \cite{Bielecki2011a}) we define $\wtd^i$s as follows:

\begin{itemize}

\item We set $\widetilde{\delta }_{{t}}^{1}=\delta _{{t}}^{1}-C_{{t}}$. This is because after the collateral transfer the counterparty pays the remaining recovery amount $\delta _{{t}}^{1}-C_{{t}}$.

\item At time ${t=\tau =\tau }_{2}$, when the counterparty defaults, then, after the collateral transfer takes place, if the uncollateralized mark-to-market (MtM) of the CDS contract, that is $S_t+ \mathbb{I}_{\{t = \tau_1 \}}\delta_t^1-C_t$ \footnote{The term ${\mathbb{I}}_{\left\{ t{=\tau }_{1}\right\} }\delta _{{t}}^{1}$ represents the exposure in case when the counterparty and the underlying entity default simultaneously.}, is negative,
    the investor closes out the position by paying the defaulting counterparty the uncollateralized MtM. If the uncollateralized MtM is positive, the investor closes out the position and receives a fraction $R_{2}$ of the uncollateralized MtM from the counterparty. Therefore, in this case, the close-out payment is defined as, \begin{equation*}
    \widetilde{\delta }_{{t}}^{2}=R_{2}\left( S_t+{\mathbb{I}}_{\left\{ t{=\tau }_{1}\right\} }\delta _{{t}}^{1}-C_{{t}}\right) ^{+}-\left( S_t+\mathbb{I}_{\{ t=\tau_1\}}\delta_t^1-C_t\right)^{-}.
    \end{equation*}

\item In case of investor default, that is at time ${t=\tau =\tau }_{3}$, if the uncollateralized MtM is positive, that is if $S_t+\mathbb{I}_{\{ t=\tau_1\}}\delta_t^1-C_t>0$, the counterparty closes out the position by paying the uncollateralized MtM. If the uncollateralized MtM is negative, the counterparty receives a fraction $R_{3}$ of the uncollateralized MtM. Hence, the close-out payment is defined as,
    \begin{equation*}
    \widetilde{\delta }_{{t}}^{3}=\left( S_{{t}}+{\mathbb{I}}_{\left\{ t{=\tau } _{1}\right\} }\delta _{{t}}^{1}-C_{{t}}\right) ^{+}-R_{3}\left( S_t+\mathbb{I}_{\{ t=\tau_1\}}\delta_t^1-C_t\right)^{-}.
    \end{equation*}

\item If the investor and the counterparty default simultaneously at time ${t=\tau =\tau }_{2}={\tau }_{{3}}$, and if the uncollateralized MtM negative, the counterparty receives a fraction $R_{3}$ of the uncollateralized MtM; however, if the uncollateralized MtM is positive, the investor receives a fraction $R_{2}$ of the uncollateralized MtM. Therefore, we set,
    \begin{equation*}
    \widetilde{\delta }_{{t}}^{4}=-\left( S_t+\mathbb{I}_{\{ t=\tau_1\}}\delta_t^1-C_t\right) .
    \end{equation*}

\item If $t=\tau=\hat{\tau}=\tau_1$, that is when the investor or the counterparty default simultaneously with the reference entity, investor receives a fraction $R_{2}$ of the remaining recovery amount, $\left( \delta _{{t}}^{1}-C_{{t}}\right) ^{+}$, when the counterparty defaults. Likewise, if the investor defaults, the counterparty receives a portion $R_{3}$ of the remaining recovery amount, $\left( \delta _{{t}}^{1}-C_{{t}}\right) ^{-}$. The close-out payment in joint defaults including the underlying entity has the form,
    \begin{equation*}
    \widetilde{\delta }_{{t}}^{5}=-\left( \delta _{{t}}^{1}-C_{{t}}\right) .
    \end{equation*}
\end{itemize}

We are now ready to define the  price processes associated with a counterparty risky CDS contract.

\begin{definition}
The ex-dividend price process $S^{C}$ of a counterparty risky CDS contract maturing at time $T$ is given as,
\begin{equation}
S_{t}^{C}=B_{t}{\mathbb{E}}_{\mathbb{Q}}\left( \left. \int_{\left] t,T\right] }B_{u}^{-1}dD_{u}^{C}\right\vert \mathcal{G}_{t}\right),\ \ t\in [0,T].  \label{pr2}
\end{equation}
The cumulative price process $\widehat{S}^{C}$ of a counterparty risky CDS contract is given by,
\begin{equation*}
\widehat{S}_{{t}}^{C}=S_{{t}}^{C}+B_{{t}}\int_{\left] {0},t\right]}B_{u}^{-1}dD_{u}^{C},\ \ t\in [0,T].
\end{equation*}

\end{definition}

\subsubsection{Bilateral Margin Agreement and Collateral Modeling}\label{marginSec}

Collateralization is one of the most important techniques of mitigation of counterparty risk, and modeling the collateral process (also termed the margin call process) is of great practical importance (cf. \cite{Algorithmics2009}). In this section, we propose a model to describe formation of the required collateral amount at every time $t\in [0,T]$, with regard to bilateral margin agreements.\footnote{ A bilateral margin agreement is a contractual agreement governed by a Credit Support Annex (CSA), which is a regulatory part of the ISDA Master Agreement (cf. \cite{ISDA2005}, page 34) describing the use of collateral which is either directly transferred between counterparties or held by a third party such as a clearing house (cf. \cite{ISDA2005}, page 68).} The following contractual parameters are essential in bilateral margin agreements and they are precisely defined in CSA documents.

\textit{Margin Period of Risk}: The margin period of risk consists of several components. Firms usually monitor their exposure on a periodic basis and receive or make appropriate margin calls considering other collateral parameters. The frequency of this process is called the margin call period and it is typically one day. This period includes a number of phases such as computation, negotiation, verification and settlement of the margin call also with possible disputes during the process. According to the ISDA Master Agreement, in case of a potential default, the defaulting counterparty enters into a short forbearance period to recover from a potential default event where the collateral is pledged by the other firm. This time interval is called the cure period. If the default is uncured, liquidation process of the collateral assets starts (cf. \cite{ISDA2010a}, page 26). This period mainly depends on the collateral portfolio selection, precise assessment of asset correlation and concentration risks as well as their liquidity, volatility and credit quality parameters. Therefore, the time interval from the last margin call plus the cure period until all collateral assets are liquidated and the resulting market risk is re-hedged is called the margin period of risk (cf. \cite{Pykhtin2009}); we shall denote it  as $\Delta $.

\textit{Threshold}: The threshold is the unsecured credit exposure that both counterparties are willing to tolerate without holding any collateral. Bilateral margin agreements specify thresholds for both counterparties and require them to post collateral whenever the current credit exposure exceeds their thresholds (cf. \cite{ISDA2010a}, page 11). These threshold amounts are defined in the related CSA documents and often set to react to the changes in the credit rating of the counterparties (cf.\cite{ISDA2010}, page 13). We will denote the counterparty and the investor's thresholds by $\Gamma _{cpty}$ and $\Gamma _{inv}$, respectively. Since we are doing our analysis from the point of view of the investor, we set the counterparty's threshold $\Gamma _{cpty}$ to be a non-negative constant, and the investor's threshold $\Gamma _{inv}$ to be a non-positive constant.

\textit{Minimum Transfer Amount}: Margin calls for amounts smaller than the $MTA$ are not allowed. The purpose of the $MTA$ is to prevent calling small amounts; this is done so to avoid the operational costs involved in small transactions (cf. \cite{ISDA2010a}, page 13). Minimum transfer amount is usually the same for the investor and the counterparty.

\textit{Re-hypothecation Risk and Segregation}: Collateral assets can be reused as a funding source on another derivatives transactions. This is known as rehypothecation. An investor (counterparty) can rehypothecate the collateral received from the counterparty (investor) by selling or lending out the assets to a third party, which dramatically increases the credit risk associated with the collateral. Elimination of this rehypothecation risk is essentially done by segregating the collateral to a third party, such as a clearing house. This procedure carries certain funding risks, since the counterparties will not be getting funding benefit from the collateral posted, so they need to raise funding in connection with their transactions using their own funding rates.

The construction of the collateral process presented below builds upon the construction given in \cite{Bielecki2011a}.
Let us denote the margin call dates by $0<t_1 < \dots < t_n< T$. On each margin call date, if the exposure is above the counterparty's threshold, $\Gamma_{cpty}$, and the difference between the current exposure and the collateral amount is greater than the MTA the counterparty posts collateral and updates the margin account; otherwise, no collateral exchange takes place since the transfer amount is less than the MTA. Likewise, the investor delivers collateral on each margin call date, if the exposure is below investor's threshold, $\Gamma_{inv}$, and the difference between the current exposure and the collateral amount is greater than MTA (cf. \cite{ISDA2005}, pages 52-56). Note that in this model a collateral transfer are allowed only if it is greater than the $MTA$ amount.

In accordance with the above discussion the collateral process is modeled as follows:

We set $C_0=0$. Then, for $i=1,2,\ldots,n$ we postulate that
\begin{align*}
C_{t_i+} &= {\mathbb{I}}_{\{S_{t_i}-\Gamma_{cpty}-C_{t_i}>MTA\}}(S_{t_i}-\Gamma_{cpty}-C_{t_i}) \\
& \ \ \ +{\mathbb{I}}_{\left\{ S_{t_i}-\Gamma_{inv}-C_{t_i}<-MTA\right\} }( S_{t_i}-\Gamma _{inv}-C_{t_i}) + C_{t_i} ,
\end{align*}
on the set $\left\{ t_i<{\hat{\tau} }\right\}$, and it is constant on interval $(t_i,t_{i+1}]$. Moreover, $C_{t} = C_{\hat{\tau}}$ on the set $\{ \hat{\tau} < t < \hat{\tau} + \Delta \}$.

Observe that the collateral increments at each margin call date $t_i<{\hat{\tau}}$ can now be represented as,
\begin{align*}
\Delta C_{t_i}:&=C_{t_i+} - C_{t_i}  \\
&= {\mathbb{I}}_{\{S_{t_i}-\Gamma_{cpty}-C_{t_i}>MTA\}}(S_{t_i}-\Gamma_{cpty}-C_{t_i}) \\
& \ \ \ +{\mathbb{I}}_{\left\{ S_{t_i}-\Gamma_{inv}-C_{t_i}<-MTA\right\} }( S_{t_i}-\Gamma _{inv}-C_{t_i})
\end{align*}

One should also note that the collateral construction given in \cite{Pykhtin2009}, which reads
\begin{equation*}
C_{t}={\mathbb{I}}_{\left\{ S_{{t}}>\Gamma _{cpty}+MTA\right\} }\left( S_{{t}
}-\Gamma _{cpty}\right) +{\mathbb{I}}_{\left\{ S_{{t}}<\Gamma
_{buy}-MTA\right\} }\left( S_{{t}}-\Gamma _{inv}\right),
\end{equation*}
allows intermediate collateral updates that are smaller than MTA. In our case, we avoid this intricacy by defining the collateral process as a left-continuous, piecewise constant process.

\begin{remark}

The collateral construction described above is cash based. The net cash value of the collateral portfolio is determined using haircuts.

The haircut (or, valuation percentage) describes the amount that will be charged from a particular collateral asset. Effective value of the collateral asset is determined by subtracting the mark-to-market value of the asset multiplied by an appropriate haircut (cf. \cite{ISDA2005}, page 67). Therefore, the haircuts applied to collateral assets should reflect the market risk on those assets. The haircut is defined as a percentage, where 0\% haircut implies complete mark-to-market value of the asset to be used as collateral without any discounting. Government securities having high credit rating such as Treasury bonds and Treasury bills are usually subjected to 1\% to 10\% haircut, while for more risky, volatile or illiquid securities, such as a stock option, the haircut might be as high as 30\%. The only asset that is not subjected to any haircut as collateral is cash where usually both parties mutually agree the use of an overnight index rate (cf. \cite{ISDA2010a}, page 27). The term valuation percentage is also used in Credit Support Annex (CSA) documents. The valuation percentage defines the amount that the market value of the asset is multiplied by to yield the effective collateral value of the asset. Hence, we have $\textrm{VP}_{t}=1-h_{t}$, where $\textrm{VP}_{t}$ is the valuation percentage and $h_{t}$ is the total haircut applied to the collateral assets at time $t$. We will not go into the details of the formation of the haircut since it is either pre-determined in the CSA documents or related to market risk measures such as VaR of the collateral assets. (cf. \cite{ISDA2005}, page 68). The main purpose of the haircut is to mitigate amortization or depreciation in the collateral asset value at the time of a default and in the margin period of risk. Moreover, the haircut should be updated as frequently as possible to reflect the changes in the volatility or liquidity of the collateral assets (cf. \cite{ISDA2005}, page 63).

Therefore, the total value of the collateral portfolio at time $t$ is equal to $(1+h_{t})C_{t}$, where $h_{t}$ is the appropriate haircut applied to the collateral portfolio.

\end{remark}

\subsection{Bilateral Credit Valuation Adjustment}

In this section, we shall compute the CVA on a CDS contract, subject to a bilateral margin agreement.

\begin{definition}
The bilateral Credit Valuation Adjustment process on a CDS contract maturing at time $T$ is defined as
\begin{equation}
\text{CVA}_{t}= S_{t}-S_{t}^{C},
\label{CVAdef}
\end{equation}
for every $t\in [0,T]$.
\end{definition}

We now present an alternative representation for the bilateral CVA, which is convenient for computational purposes.

\begin{proposition}\label{prop1}
The bilateral CVA process on a CDS contract maturing at time $T$ satisfies
\begin{align}
\text{CVA}_{t} &=B_{t}{\mathbb{E}}_{\mathbb{Q}}\left( \left. {\mathbb{I}}
_{\left\{ t<{\tau =\tau }_{2}\leq T\right\} }B_{{\tau }}^{-1}\left(
1-R_{2}\right) \left( S_{{\tau }}+{\mathbb{I}}_{\left\{ {\tau =\tau }
_{1}\right\} }\delta _{{\tau }}^{1}-C_{{\tau }}\right) ^{+}\right\vert
\mathcal{G}_{t}\right)  \notag \\
& \ \ \ -B_{t}{\mathbb{E}}_{\mathbb{Q}}\left( \left. {\mathbb{I}}_{\left\{ t<{\tau
=\tau }_{3}\leq T\right\} }B_{{\tau }}^{-1}\left( 1-R_{3}\right) \left( S_{{
\tau }}+{\mathbb{I}}_{\left\{ {\tau =\tau }_{1}\right\} }\delta _{{\tau }
}^{1}-C_{{\tau }}\right) ^{-}\right\vert \mathcal{G}_{t}\right), \label{CVA}
\end{align}
for every $t\in [0,T]$.
\end{proposition}

\begin{proof}
We begin by observing that
\begin{equation*}
\int_{\left] {t},T\right] }B_{u}^{-1}\widetilde{\delta }_{u}^{i}\left(
1-H_{u-}\right) dH_{u}^{i}=B_{{\tau
}}^{-1}\widetilde{\delta }_{{\tau }}^{i}{\mathbb{I}}_{\left\{ t<{\tau =\tau }
_{i}\leq T\right\} },
\end{equation*}
for $i=1,2,3$. Consequently,
\begin{align}
\int_{\left] {t},T\right] }B_{u}^{-1}dD_{u}^{C} &=B_{{\tau }}^{-1}
\widetilde{\delta }_{{\tau }}^{1}{\mathbb{I}}_{\left\{ t<{\tau =\tau }
_{1}\leq T\right\} }+B_{{\tau }}^{-1}\widetilde{\delta }_{{\tau }}^{2}{
\mathbb{I}}_{\left\{ t<{\tau =\tau }_{2}\leq T\right\} } \notag \\
& \ \ \ +B_{{\tau }}^{-1}\widetilde{\delta }_{{\tau }}^{3}{\mathbb{I}}_{\left\{ t<{\tau =\tau }
_{3}\leq T\right\} } +B_{{\tau }}^{-1}\widetilde{\delta }_{{\tau }}^{4}{\mathbb{I}}_{\left\{ t<{
\tau =\tau }_{2}={\tau }_{3}\leq T\right\} } \notag \\
& \ \ \ +B_{{\tau }}^{-1}\widetilde{ \delta }_{{\tau }}^{5}{\mathbb{I}}_{\left\{ t<{\tau =\tau }^{\ast }={\tau }
_{1}\leq T\right\} } +B_{{\tau }}^{-1}C_{{\tau }}{\mathbb{I}}_{\left\{ t<{
\tau }\leq T\right\} } \notag \\
& \ \ \ -\kappa \int_{\left] t,T\right] }B_{u}^{-1}{\mathbb{I}}_{\left\{ {\tau }
>u\right\} }du \label{F0}.
\end{align}
Using the definitions of the close-out cash-flows $\widetilde{\delta }_{{\tau }}^{i}, \ i=1,\ldots,5,$ we get from \eqref{F0}
\begin{align}
\int_{\left] {t},T\right] }B_{u}^{-1}dD_{u}^{C} &=B_{{\tau }}^{-1}\left(
\delta _{{\tau }}^{1}-C_{{\tau }}\right) {\mathbb{I}}_{\left\{ t<{\tau =\tau
}_{1}\leq T\right\} } \notag \\
& \ \ \ -\kappa \int_{\left] t,T\right] }B_{u}^{-1}{\mathbb{I}}
_{\left\{ {\tau }>u\right\} }du+B_{{\tau }}^{-1}C_{{\tau }}{\mathbb{I}}
_{\left\{ t<{\tau }\leq T\right\} }  \label{F1} \\
& \ \ \ +B_{{\tau }}^{-1}\left( R_{2}\left( S_{{\tau }}\mathbb{+}{\mathbb{I}}
_{\left\{ {\tau =\tau }_{1}\right\} }\delta _{{\tau }}^{1}-C_{{\tau }
}\right) ^{+} \right. \notag \\
& \ \ \ \left.-\left( S_{{\tau }}+{\mathbb{I}}_{\left\{ {\tau =\tau }
_{1}\right\} }\delta _{{\tau }}^{1}-C_{{\tau }}\right) ^{-}\right) {\mathbb{I
}}_{\left\{ t<{\tau =\tau }_{2}\leq T\right\} }  \notag \\
& \ \ \ +B_{{\tau }}^{-1}\left( \left( S_{{\tau }}+{\mathbb{I}}_{\left\{ {\tau
=\tau }_{1}\right\} }\delta _{{\tau }}^{1}-C_{{\tau }}\right)
^{+} \right. \notag \\
& \ \ \ \left. -R_{3}\left( S_{{\tau }}+{\mathbb{I}}_{\left\{ {\tau =\tau }_{1}\right\}
}\delta _{{\tau }}^{1}-C_{{\tau }}\right) ^{-}\right) {\mathbb{I}}_{\left\{
t<{\tau =\tau }_{3}\leq T\right\} }  \notag \\
& \ \ \ -B_{{\tau }}^{-1}\left( S_{{\tau }}+{\mathbb{I}}_{\left\{ {\tau =\tau }
_{1}\right\} }\delta _{{\tau }}^{1}-C_{{\tau }}\right) {\mathbb{I}}_{\left\{
t<{\tau =\tau }_{2}={\tau }_{3}\leq T\right\} } \notag \\
& \ \ \ -B_{{\tau }}^{-1}\left(
\delta _{{\tau }}^{1}-C_{{\tau }}\right) {\mathbb{I}}_{\left\{ t<{\tau =\tau
}^{\ast }={\tau }_{1}\leq T\right\} }  \notag.
\end{align}
Since
\begin{align*}
\mathbb{I}_{\left\{ t<{\tau }\leq T\right\} } &= {\mathbb{I}}_{\left\{
t<{\tau =\tau }_{1}\leq T\right\} }+{\mathbb{I}}_{\left\{ t<{\tau =\tau }
_{2}\leq T\right\} }+{\mathbb{I}}_{\left\{ t<{\tau =\tau }_{3}\leq T\right\}} \\
& \ \ \ -{\mathbb{I}}_{\left\{ t<{\tau =\tau }_{2}={\tau }_{3}\leq T\right\} }-{
\mathbb{I}}_{\left\{ t<{\tau }^{{\ast }}{{=\tau }_{1}}\leq T\right\} },
\end{align*}
using the equality
\begin{equation*}
R_{i}\left( S_{{\tau }}-C_{{\tau }}\right) ^{+}-\left( S_{{\tau }}-C_{{\tau }
}\right) ^{-}+C_{{\tau }}=S_{{\tau }}-\left( 1-R_{i}\right) \left( S_{{\tau }
}-C_{{\tau }}\right) ^{+}
\end{equation*}
and observing that ${\mathbb{I}}_{\left\{ {\tau =\tau }_{1}\right\} }S_{{ \tau }}=0,$ we can rearrange the terms in $\left( \ref{F1}\right) $ as follows,
\begin{align}
\int_{\left] {t},T\right] }B_{u}^{-1}dD_{u}^{C} &=B_{{\tau }}^{-1}\delta _{{
\tau }}^{1}{\mathbb{I}}_{\left\{ t<{\tau =\tau }_{1}\leq T\right\} }-\kappa
\int_{\left] t,T\right] }B_{u}^{-1}{\mathbb{I}}_{\left\{ {\tau }>u\right\}
}du  \label{F2} \\
& \ \ \ +B_{{\tau }}^{-1}S_{{\tau }}\left( {\mathbb{I}}_{\left\{ t<{\tau =\tau }
_{2}\leq T\right\} }+{\mathbb{I}}_{\left\{ t<{\tau =\tau }_{3}\leq T\right\}
} \right. \notag \\
& \ \ \ \left. -{\mathbb{I}}_{\left\{ t<{\tau =\tau }_{2}={\tau }_{3}\leq T\right\}
}\right) {\mathbb{I}}_{\left\{ {\tau \neq \tau }_{1}\right\} }  \notag \\
& \ \ \ -B_{{\tau }}^{-1}\left( 1-R_{2}\right) \left( S_{{\tau }}+{\mathbb{I}}
_{\left\{ {\tau =\tau }_{1}\right\} }\delta _{{\tau }}^{1}-C_{{\tau }
}\right) ^{+}{\mathbb{I}}_{\left\{ t<{\tau =\tau }_{2}\leq T\right\} }
\notag \\
& \ \ \ +B_{{\tau }}^{-1}\left( 1-R_{3}\right) \left( S_{{\tau }}+{\mathbb{I}}
_{\left\{ {\tau =\tau }_{1}\right\} }\delta _{{\tau }}^{1}-C_{{\tau }
}\right) ^{-}{\mathbb{I}}_{\left\{ t<{\tau =\tau }_{3}\leq T\right\} }.
\notag
\end{align}
Now, combining \eqref{F2} with $\eqref{d1}$ we see that
\begin{align*}
S_{{t}}^{C} &=B_{t}{\mathbb{E}}_{\mathbb{Q}}\left( \left. \left( {\mathbb{I}
}_{\left\{ t<{\tau =\tau }_{1}\leq T\right\} }+{\mathbb{I}}_{\left\{ {\tau >}
T\right\} }\right) \int_{\left] t,T\right] }B_{u}^{-1}dD_{u}\right\vert
\mathcal{G}_{t}\right) \\
& \ \ \ +B_{t}{\mathbb{E}}_{\mathbb{Q}}\left( \left. \left( {\mathbb{I}}
_{\left\{ t<{\tau =\tau }_{2}\leq T\right\} }+{\mathbb{I}}_{\left\{ t<{\tau
=\tau }_{3}\leq T\right\} } \right.\right.\right. \\
& \ \ \ \left.\left.\left.\left.  -{\mathbb{I}}_{\left\{ t<{\tau =\tau }_{2}={\tau }
_{3}\leq T\right\} }\right) {\mathbb{I}}_{\left\{ {\tau \neq \tau }
_{1}\right\} }\right) {\mathbb{E}}_{\mathbb{Q}}\left( \left. \int_{\left] {
\tau },T\right] }B_{u}^{-1}dD_{u}\right\vert \mathcal{G}_{{\tau }}\right)
\right\vert \mathcal{G}_{t}\right) \\
& \ \ \ -B_{t}{\mathbb{E}}_{\mathbb{Q}}\left( \left. {\mathbb{I}}_{\left\{ t<{\tau
=\tau }_{2}\leq T\right\} }B_{{\tau }}^{-1}\left( 1-R_{2}\right) \left( S_{{
\tau }}+{\mathbb{I}}_{\left\{ {\tau =\tau }_{1}\right\} }\delta _{{\tau }
}^{1}-C_{{\tau }}\right) ^{+}\right\vert \mathcal{G}_{t}\right) \\
& \ \ \ +B_{t}{\mathbb{E}}_{\mathbb{Q}}\left( \left. {\mathbb{I}}_{\left\{ t<{\tau
=\tau }_{3}\leq T\right\} }B_{{\tau }}^{-1}\left( 1-R_{3}\right) \left( S_{{
\tau }}+{\mathbb{I}}_{\left\{ {\tau =\tau }_{1}\right\} }\delta _{{\tau }
}^{1}-C_{{\tau }}\right) ^{-}\right\vert \mathcal{G}_{t}\right) .
\end{align*}
From here, observing that
\begin{equation*}
{\mathbb{I}}_{\left\{ {\tau \leq t}\right\} }+{\mathbb{I}}_{\left\{ {\tau >}
T\right\} }+{\mathbb{I}}_{\left\{ t<{\tau =\tau }_{1}\leq T\right\} }+\left(
{\mathbb{I}}_{\left\{ t<{\tau =\tau }_{2}\leq T\right\} }+{\mathbb{I}}
_{\left\{ t<{\tau =\tau }_{3}\leq T\right\} }-{\mathbb{I}}_{\left\{ t<{\tau
=\tau }_{2}={\tau }_{3}\leq T\right\} }\right) {\mathbb{I}}_{\left\{ {\tau
\neq \tau }_{1}\right\} }=1,
\end{equation*}
we get
\begin{align}
S_{{t}}^{C} &=B_{t}{\mathbb{E}}_{\mathbb{Q}}\left( \left. \int_{\left] t,T
\right] }B_{u}^{-1}dD_{u}\right\vert \mathcal{G}_{t}\right)  \notag \\
& \ \ \ -B_{t}{\mathbb{E}}_{\mathbb{Q}}\left( \left. {\mathbb{I}}_{\left\{ t<{\tau
=\tau }_{2}\leq T\right\} }B_{{\tau }}^{-1}\left( 1-R_{2}\right) \left( S_{{
\tau }}-C_{{\tau }}\right) ^{+}\right\vert \mathcal{G}_{t}\right)  \notag \\
& \ \ \ +B_{t}{\mathbb{E}}_{\mathbb{Q}}\left( \left. {\mathbb{I}}_{\left\{ t<{\tau
=\tau }_{3}\leq T\right\} }B_{{\tau }}^{-1}\left( 1-R_{3}\right) \left( S_{{
\tau }}-C_{{\tau }}\right) ^{-}\right\vert \mathcal{G}_{t}\right),  \notag \\
\end{align}
which is
\begin{align}
S_{{t}}^{C} &= S_t-B_{t}{\mathbb{E}}_{\mathbb{Q}}\left( \left. {\mathbb{I}}_{\left\{ t<{\tau
=\tau }_{2}\leq T\right\} }B_{{\tau }}^{-1}\left( 1-R_{2}\right) \left( S_{{
\tau }}-C_{{\tau }}\right) ^{+}\right\vert \mathcal{G}_{t}\right)  \notag \\
& \ \ \ +B_{t}{\mathbb{E}}_{\mathbb{Q}}\left( \left. {\mathbb{I}}_{\left\{ t<{\tau
=\tau }_{3}\leq T\right\} }B_{{\tau }}^{-1}\left( 1-R_{3}\right) \left( S_{{
\tau }}-C_{{\tau }}\right) ^{-}\right\vert \mathcal{G}_{t}\right)  \notag .
\end{align}
This proves the result.
\end{proof}

\begin{remark} \label{CVAoption}
The above results shows that the value of the bilateral CVA is the same as the sum of the value of a long position in a zero-strike call option on the uncollateralized amount and the value of a short position in a zero-strike put option on the uncollateralized amount.
\end{remark}

\subsubsection{Unilateral CVA and Debt Value Adjustment}

The bilateral nature of the counterparty risk is a consequence of possible default of the counterparty and the possible default of the investor. The values of potential losses associated with these two components are called unilateral CVA (UCVA) and debt value adjustment (DVA), respectively, and defined below.

\begin{definition}
The Unilateral Credit Value Adjustment is defined as,
\begin{equation*}
\text{UCVA}_{t}=B_{t}{\mathbb{E}}_{\mathbb{Q}}\left( \left. {\mathbb{I}}
_{\left\{ t<{\tau =\tau }_{2}\leq T\right\} }B_{{\tau }}^{-1}\left(
1-R_{2}\right) \left( S_{{\tau }}+{\mathbb{I}}_{\left\{ {\tau =\tau }
_{1}\right\} }\delta _{{\tau }}^{1}-C_{{\tau }}\right) ^{+}\right\vert
\mathcal{G}_{t}\right) ,\ t\in \left[ 0,T\right],
\end{equation*}
and symmetrically the Debt Value Adjustment is defined as,
\begin{equation*}
\text{DVA}_{t}=B_{t}{\mathbb{E}}_{\mathbb{Q}}\left( \left. {\mathbb{I}}
_{\left\{ t<{\tau =\tau }_{3}\leq T\right\} }B_{{\tau }}^{-1}\left(
1-R_{3}\right) \left( S_{{\tau }}+{\mathbb{I}}_{\left\{ {\tau =\tau }
_{1}\right\} }\delta _{{\tau }}^{1}-C_{{\tau }}\right) ^{-}\right\vert
\mathcal{G}_{t}\right)  ,\ t\in \left[ 0,T\right].
\end{equation*}

\end{definition}

\begin{remark}
 DVA accounts for  the risk of investor's own default, and it represents the value of any potential outstanding liabilities of the investors that will not be honored at the time of the investor's default:

 In fact, at time of his/her default, the investor only pays to the counterparty the recovery amount, that is $R_{3}\left( S_{{\tau }}+{\mathbb{I}}_{\left\{ {\tau =\tau }_{1}\right\} }\delta _{{\tau }}^{1}-C_{{\tau } }\right) ^{-}$. Therefore, the investor gains the remaining amount, which is equal to $\left( 1-R_{3}\right) \left( S_{{\tau }}+{ \mathbb{I}}_{\left\{ {\tau =\tau }_{1}\right\} }\delta _{{\tau }}^{1}-C_{{ \tau }}\right) ^{-},$ on his/her outstanding liabilities by defaulting.  Risk management of this component is of great importance for financial institutions.

 When considering the unilateral counterparty risk DVA is set to zero.

\end{remark}

In view of Proposition \ref{prop1} and of the above definition we have that
\begin{equation*}
\text{CVA}_{t}=\text{UCVA}_{t}-\text{DVA}_{t},\ t\in [ 0,T].
\end{equation*}

Note that the bilateral CVA amount may be negative for the investor due to ``own default risk''. This also indicates that the price $S^{C}$ of counterparty risky CDS contract may be greater than the price $S$ of counterparty risk-free contract.

\begin{remark} \textbf{(Upfront CDS Conversion)}

After the ``CDS Big Bang'' (cf. \cite{Markit2009}) a process has been originated to replace standard CDS contracts with so called {\it upfront CDS contracts}.
An upfront CDS contract is composed of an upfront payment, which is an amount to be exchanged upon the inception of the
contract, and of a fixed spread.
The fixed spread, say $\widehat{\kappa }$, will be 100bps for investment grade CDS contracts, and 500bps for high yield CDS contracts.
The recovery rate is also standardized to two possible values: 20\% or 40\%, depending on the credit worthiness of the reference name.
The corresponding cumulative dividend process of a counterparty-risk-free CDS contract is described in the following definition.

\begin{definition}\label{def:CumDivUpfrontCDS}
The cumulative dividend process $\widehat{D}$ of a counterparty-risk-free upfront CDS contract, maturing at time $T$, is given as
\begin{equation*}
\widehat{D}_{t}=\int_{\left] 0,t\right] }\delta _{u}^{1}dH_{u}^{1}- \textrm{UP} -
\widehat{\kappa }\int_{\left] 0,t\right] }\left( 1-H_{u}^{1}\right) du \, , \quad t\in \left[ 0,T\right],
\end{equation*}
where $\textrm{UP}$ is the upfront payment, and $\widehat{\kappa }$ is the fixed spread.
\end{definition}

Recall that the spread $\kappa_0$ of a standard CDS contract is set such that the protection leg $PL_0$ and fixed leg $\kappa_0 DV01_0$ are equal at initiation (making  the price of the contract to be zero). Similarly, in the case of an upfront CDS contract, with $\widehat{\kappa}$ being fixed, the upfront payment $\textrm{UP}$ is chosen such that the contract has zero value at initiation. It is easy to convert the conventional spread $\kappa_0$ into an upfront payment PU and vise versa. Indeed, directly from the Definition \ref{def:CumDivUpfrontCDS}, and  definitions of $PL_0$ and $DV01_0$, we have
\begin{equation*}
PL_{0}-UP-\widehat{\kappa }DV01_{0}=PL_{0}-\kappa _{0}DV01_{0}=0\, ,
\end{equation*}
which implies the following representations
\begin{equation*}
UP=\left( \kappa _{0}-\widehat{\kappa }\right) DV01_{0}, \qquad
\kappa _{0}=\frac{UP}{DV01_{0}}+\widehat{\kappa }\, .
\end{equation*}

In view of the conversion formulae presented above the discussion of CVA, DVA and UCVA done for standard CDS contracts can be adopted to the case of the upfront CDS contracts in a straightforward manner.

\end{remark}

\subsubsection{CVA via Credit Exposures}

Credit exposure is defined as the potential loss that may be suffered by either one of the counterparties due to the other party's default. Here, we discuss some measures commonly used to quantify credit exposure, such as \textit{Potential Future Exposure} (PFE), \textit{Expected Positive Exposure} (EPE) and \textit{Expected Negative Exposure} (ENE), and their relation to CVA.

\medskip
\noindent
Potential Future Exposure is the basic measure of credit exposure:
\begin{definition}
Potential Future Exposure of a CDS contract with a bilateral margin agreement is defined as follows,
\begin{align}
PFE &= {\mathbb{I}}_{\left\{ {\tau =\tau }_{2}\right\} }\left( 1-R_{2}\right)
\left( S_{{\tau }}+{\mathbb{I}}_{\left\{ {\tau =\tau }_{1}\right\} }\delta _{
{\tau }}^{1}-C_{{\tau }}\right) ^{+} \notag \\
& \ \ \ -{\mathbb{I}}_{\left\{ {\tau =\tau }_{3}\right\} }\left( 1-R_{3}\right) \left( S_{{\tau }}+{\mathbb{I}}_{\left\{
{\tau =\tau }_{1}\right\} }\delta _{{\tau }}^{1}-C_{{\tau }}\right) ^{-}. \notag
\end{align}
\end{definition}
Note that there exists several forms in which the potential future exposure is defined by financial institutions. The PFE definition given above, as a random variable, is in line with the PFE definitions in (cf. \cite{DePriscoRosen2005}), as opposed to the rather classical definition of the PFE as the quantile of the exposure distribution (cf. \cite{Cesari2010}).
\begin{remark}
Observe that the CVA is related to PFE as follows,
\begin{equation*}
CVA_{t}=B_{t}{\mathbb{E}}_{\mathbb{Q}}\left( \left. {\mathbb{I}}_{\left\{ t<{
\tau }\leq T\right\} }B_{{\tau }}^{-1}PFE\right\vert \mathcal{G}_{t}\right),\ t\in \left[ 0,T\right].
\end{equation*}

\end{remark}

Expected Positive Exposure is defined as the expected amount the investor will lose if the counterparty default happens at time $t$, and Expected Negative Exposure is defined as the expected amount the investor will lose if his own default happens at time $t$. Note that the losses are conditional on default at time $t$. EPE and ENE are necessary quantities to price and hedge counterparty risk.

\begin{definition}
The Expected Positive Exposure of a CDS contract with a bilateral margin agreement is defined as,
\begin{equation*}
EPE_{t}={\mathbb{E}}_{\mathbb{Q}}\left( \left. \left( 1-R_{2}\right) \left(
S_{{\tau }}+{\mathbb{I}}_{\left\{ {\tau =\tau }_{1}\right\} }\delta _{{\tau }
}^{1}-C_{{\tau }}\right) ^{+}\right\vert {\tau }={\tau }_{2}=t\right) ,
\end{equation*}
and the Expected Negative Exposure is defined as,
\begin{equation*}
ENE_{t}={\mathbb{E}}_{\mathbb{Q}}\left( \left. \left( 1-R_{3}\right) \left(
S_{{\tau }}+{\mathbb{I}}_{\left\{ {\tau =\tau }_{1}\right\} }\delta _{{\tau }
}^{1}-C_{{\tau }}\right) ^{-}\text{ }\right\vert {\tau }={\tau }_{3}=t\right)
\end{equation*}
for every $t\in[0,T]$.
\end{definition}

\subsection{Dynamics of CVA}\label{sec:CVAdynamics}
In this section we derive the  dynamics for CVA. This is important for deriving formulae for dynamic hedging of counterparty risk, the issue that will be discussed in a different paper.

We begin with defining some auxiliary stopping times, that will come handy later on:

\begin{eqnarray*}
&&\tau ^{\{1\}} :=
\begin{cases}
\tau _{1} & \text{if $\tau _{1}\neq \tau _{2},\,\tau _{1}\neq \tau _{3}$} \\
\infty & \text{otherwise}
\end{cases}
,\text{ }\tau ^{\{2\}} :=
\begin{cases}
\tau _{2} & \text{if }{\tau }_{{2}}\neq {\tau }_{1},{\tau }_{{2}}\neq {\tau }
_{{3}} \\
\infty & \text{otherwise}
\end{cases}
, \\
&&\tau ^{\{3\}} :=
\begin{cases}
\tau _{3} & \text{if }{\tau }_{{3}}\neq {\tau }_{{1}},{\tau }_{{3}}\neq {
\tau }_{{2}} \\
\infty & \text{otherwise}
\end{cases}
,\text{ }\tau ^{\{4\}} :=
\begin{cases}
\tau _{2} & \text{if $\tau _{2} =\tau _{3},\,\tau _{2}\neq \tau _{1}$} \\
\infty & \text{otherwise}
\end{cases}
, \\
&&\tau ^{\{5\}} :=
\begin{cases}
\tau _{1} & \text{if $\tau _{1} =\tau _{2},\,\tau _{1}\neq \tau _{3}$} \\
\infty & \text{otherwise}
\end{cases}
,\text{ }\tau ^{\{6\}} :=
\begin{cases}
\tau _{1} & \text{if $\tau _{1} =\tau _{3},\,\tau _{1}\neq \tau $}_{2} \\
\infty & \text{otherwise}
\end{cases}
,
\end{eqnarray*}
$$\tau ^{\{7\}} :=
\begin{cases}
\tau _{1} & \text{if $\tau _{1} =\tau _{2} =\tau _{3}$} \\
\infty & \text{otherwise}
\end{cases}
.$$

Accordingly, we define the default indicator processes:
\begin{align*}
H_{t}^{\left\{ 1\right\} }:&={\mathbb{I}}_{\left\{ {\tau }_{{1}}\leq t,{
\tau }_{{1}}\neq {\tau }_{{2}},{\tau }_{{1}}\neq {\tau }_{{3}}\right\} }={
\mathbb{I}}_{\left\{ {\tau }^{\left\{ 1\right\} }\leq t,\right\} },&
H_{t}^{\left\{ 2\right\} }:&={\mathbb{I}}_{\left\{ {\tau }_{{2}}\leq t,{\tau }
_{{2}}\neq {\tau }_{1},{\tau }_{{2}}\neq {\tau }_{{3}}\right\} }={\mathbb{I}}
_{\left\{ {\tau }^{\left\{ 2\right\} }\leq t,\right\} }, \\
H_{t}^{\left\{ 3\right\} }:&={\mathbb{I}}_{\left\{ {\tau }_{{3}}\leq t,{
\tau }_{{3}}\neq {\tau }_{{1}},{\tau }_{{3}}\neq {\tau }_{{2}}\right\} }={
\mathbb{I}}_{\left\{ {\tau }^{\left\{ 3\right\} }\leq t,\right\} },&
H_{t}^{\left\{ 4\right\} }:&={\mathbb{I}}_{\left\{ {\tau }_{{2}}={\tau }_{{3}
}\leq t,{\tau }_{{1}}\neq {\tau }_{{2}}\right\} }={\mathbb{I}}_{\left\{ {
\tau }^{\left\{ 4\right\} }\leq t,\right\} }, \\
H_{t}^{\left\{ 5\right\} }:&={\mathbb{I}}_{\left\{ {\tau }_{{1}}={\tau }_{{2
}}\leq t,{\tau }_{{1}}\neq {\tau }_{{3}}\right\} }={\mathbb{I}}_{\left\{ {
\tau }^{\left\{ 5\right\} }\leq t,\right\} },&
H_{t}^{\left\{ 6\right\} }:&={\mathbb{I}}_{\left\{ {\tau }_{{1}}={\tau }_{{3}}\leq t,{\tau }_{{1}}\neq {
\tau }_{{2}}\right\} }={\mathbb{I}}_{\left\{ {\tau }^{\left\{ 6\right\}
}\leq t,\right\} },
\end{align*}
$$H_{t}^{\left\{ 7\right\} }:={\mathbb{I}}_{\left\{ {\tau }_{{1}}={\tau }
_{{2}}={\tau }_{{3}}\leq t\right\} }={\mathbb{I}}_{\left\{ {\tau }^{\left\{
7\right\} }\leq t\right\} }.$$

\begin{remark}
Note that one can represent  processes $H_{t}^{\left\{ i\right\}},i=1,\ldots, 7$, as follows
\begin{align*}
H_{t}^{\left\{ 7\right\} }&=\left[ \left[ H^{1},H^{2}\right] ,H^{3}\right]_{t}, \; \, \, \: \,
H_{t}^{\left\{ 6\right\} }=\left[ H^{1},H^{3}\right] _{t}-H_{t}^{\left\{ 7\right\} }, \\
H_{t}^{\left\{ 5\right\} }&=\left[ H^{1},H^{2}\right] _{t}-H_{t}^{\left\{ 7\right\} }, \, \; \,
H_{t}^{\left\{ 4\right\} }=\left[ H^{2},H^{3}\right] _{t}-H_{t}^{\left\{ 7\right\} }, \\
H_{t}^{\left\{ 3\right\} }&=H_{t}^{3}-H_{t}^{\left\{ 4\right\} }-H_{t}^{\left\{ 6\right\}}-H_{t}^{\left\{ 7\right\} }, \\
H_{t}^{\left\{ 2\right\} }&=H_{t}^{2}-H_{t}^{\left\{ 4\right\} }-H_{t}^{\left\{ 5\right\}}-H_{t}^{\left\{ 7\right\} }, \\
H_{t}^{\left\{ 1\right\} }&=H_{t}^{1}-H_{t}^{\left\{ 5\right\} }-H_{t}^{\left\{ 6\right\}}-H_{t}^{\left\{ 7\right\} }.
\end{align*}
In particular, these processes are $\mathbb{G}$-adapted processes.
\end{remark}

Let $G(t) =\mathbb{Q}\left( \left. {\tau }>t\right\vert \mathcal{F}_{t}\right) $ be the survival probability process of ${\tau }$ with respect to filtration $\mathbb{F}$. It is a $\bF$ supermartingale and it admits unique Doob-Meyer decomposition $G=\mu - \nu$ where $\mu$ is the martingale part and $\nu$ is a predictable increasing process.  We assume that $G$ is a continuous process and $v$ is absolutely continuous with respect to the Lebesgue measure, so that $d\nu_{t}=v _{t}dt$ for some $\mathbb{F}$-progressively measurable, non-negative process $v$. We denote by $l$ the $\bF$-progressively measurable process defined as $l_{t}=G(t)^{-1}v _{t}$. Finally, we assume that all $\bF$ martingales are continuous.

We assume that hazard process of each stopping time ${\tau }^{\left\{i\right\} }$ admits an $( \mathbb{F},\mathbb{G})$-intensity process $q^{i}$ for every $i=1,\ldots ,7$, so that the process $M^{\left\{
i\right\} },$ given by the formula,
\begin{equation*}
M_{t}^{\left\{ i\right\} }=H_{t}^{\left\{ i\right\} }-\int_{0}^{t}\left(
1-H_{u}^{\left\{ i\right\} }\right) q_{u}^{i}du
\end{equation*}
is a $\mathbb{G}$-martingale for every $t\in \left[ 0,T\right] $ and $i=1,\ldots ,7$.

We now have the following technical result,
\begin{lemma}
The processes
\begin{equation*}
M_{t}^{i}:=M_{t\wedge \tau }^{\left\{ i\right\} }=H_{t\wedge \tau }^{\left\{ i\right\} }-\int_{0}^{t\wedge \tau
}l_{u}^{i}du,\ \ t\geq 0, i=1,2,\ldots,7,
\end{equation*}
and
\begin{equation*}
M_{t}:=H_{t\wedge \tau }-\int_{0}^{t\wedge \tau }l_{u}du,\ \ t\geq 0,
\end{equation*}
 where
\begin{equation*}
l_{t}^{i}={\mathbb{I}}_{\left\{ {\tau }\geq t\right\} }q_{t}^{i}\text{ and }
l_{t}=\sum_{i=1}^{7}l_{t}^{i}\ \ t\geq 0, i=1,2,\ldots,7,
\end{equation*}
are $\mathbb{G}$-martingales
\end{lemma}

\begin{proof}
Fix  $i=1,\ldots ,7$. Process $M^{i}$ follows a $\mathbb{G}$-martingale, since it is  $\mathbb{G}$-martingale $M^{\{i\}}$ stopped at the $\bG$ stopping time $\tau.$  Moreover, we have that
$
M_{t}=\sum_{i=1}^{7}M_{t}^{i},
$
so that process $M$ is also a $\mathbb{G}$-martingale.\end{proof}

We shall now proceed with deriving some useful representations for the processes $S^C$ and $S$.

\begin{lemma}\label{ScRep}
The ex-dividend price process, $S^{C}$, of a counterparty risky CDS contract, given in $\left( \ref{pr2}\right) $, can be represented as follows,
\begin{equation}
S_{t}^{C}=B_{t}{\mathbb{E}}_{\mathbb{Q}}\left( \left. B_{{\tau }
}^{-1}\sum_{i=1}^{7}{\mathbb{I}}_{\left\{ t<{\tau =\tau }^{\left\{ i\right\}
}\leq T\right\} }\overline{\delta }_{{\tau }}^{i}-\kappa \int_{\left] t,T
\right] \ \ }B_{u}^{-1}{\mathbb{I}}_{\left\{ {\tau }>u\right\}
}du\right\vert \mathcal{G}_{t}\right)  \label{pr2n}
\end{equation}
where
\begin{eqnarray*}
&&\overline{\delta }_{{t}}^{1}=\delta _{{t}}^{1},\text{ }\overline{\delta }_{
{t}}^{2}=S_{{t}}-\left( 1-R_{2}\right) \left( S_{{t}}-C_{t}\right) ^{+} \\
&&\overline{\delta }_{{t}}^{3}=S_{{t}}+\left( 1-R_{3}\right) \left(
S_{t}-C_{{t}}\right) ^{-}, \\
&&\overline{\delta }_{{t}}^{4}=S_{{t}}-\left( 1-R_{2}\right) \left( S_{{t}
}-C_{t}\right) ^{+}+\left( 1-R_{3}\right) \left( S_{t}-C_{{t}}\right) ^{-} \\
&&\overline{\delta }_{{t}}^{5}=\delta _{{t}}^{1}-\left( 1-R_{2}\right)
\left( \delta _{{t}}^{1}-C_{t}\right) ^{+},\text{ }\overline{\delta }_{{t}
}^{6}=\delta _{{t}}^{1}+\left( 1-R_{3}\right) \left( \delta _{{t}}^{1}-C_{{t}
}\right) ^{-} \\
&&\overline{\delta }_{{t}}^{7}=\delta _{{t}}^{1}-\left( 1-R_{2}\right)
\left( \delta _{{t}}^{1}-C_{t}\right) ^{+}+\left( 1-R_{3}\right) \left(
\delta _{{t}}^{1}-C_{{t}}\right) ^{-}.
\end{eqnarray*}
\end{lemma}

\begin{proof}
Let us rewrite $(\ref{F2})$ in the following form,
\begin{align*}
S_{t}^{C}&=B_{t}{\mathbb{E}}_{\mathbb{Q}}\left( B_{{\tau }}^{-1}\delta _{{\tau }}^{1}\sum_{i=1,5,6,7}{\mathbb{I}}_{\left\{ t<{\tau =\tau }^{\left\{i\right\} }\leq T\right\} }+B_{{\tau }}^{-1}S_{{\tau }}\sum_{i=2,3,4}{\mathbb{I}}_{\left\{ t<{\tau =\tau }^{\left\{ i\right\} }\leq T\right\}}\right. \\
& \ \ \ -B_{{\tau }}^{-1}\left( 1-R_{2}\right) \left( S_{{\tau }}+{\mathbb{I}}_{\left\{ {\tau =\tau }_{1}\right\} }\delta _{{\tau }}^{1}-C_{{\tau }}\right) ^{+}\sum_{i=2,4,5,7}{\mathbb{I}}_{\left\{ t<{\tau =\tau }^{\left\{i\right\} }\leq T\right\} } \\
& \ \ \ +B_{{\tau }}^{-1}\left( 1-R_{3}\right) \left( S_{{\tau }}+{\mathbb{I}}_{\left\{ {\tau =\tau }_{1}\right\} }\delta _{{\tau }}^{1}-C_{{\tau }}\right)^{-} \\
& \ \ \ \left.\left. \sum_{i=3,4,6,7}{\mathbb{I}}_{\left\{ t<{\tau =\tau }^{\left\{ i\right\} }\leq T\right\} }-\kappa \int_{\left]t,T\right]}B_{u}^{-1}{\mathbb{I}}_{\left\{ {\tau }>u\right\} }du\right\vert \mathcal{G}_{t}\right) ,
\end{align*}
which, after rearranging terms, leads to
\begin{align*}
S_{t}^{C}&=B_{t}{\mathbb{E}}_{\mathbb{Q}}\left( B_{{\tau }}^{-1}\delta _{{
\tau }}^{1}{\mathbb{I}}_{\left\{ t<{\tau =\tau }^{\left\{ 1\right\} }\leq
T\right\} }+B_{{\tau }}^{-1}\left( S_{{t}}-\left( 1-R_{2}\right) \left( S_{{
\tau }}-C_{\tau }\right) ^{+}\right) {\mathbb{I}}_{\left\{ t<{\tau =\tau }
^{\left\{ 2\right\} }\leq T\right\} }\right. \\
& \ \ \ +B_{{\tau }}^{-1}\left( S_{{\tau }}+\left( 1-R_{3}\right) \left( S_{\tau
}-C_{{\tau }}\right) ^{-}\right) {\mathbb{I}}_{\left\{ t<{\tau =\tau }
^{\left\{ 3\right\} }\leq T\right\} } \\
& \ \ \ +B_{{\tau }}^{-1}\left( S_{{\tau }}-\left( 1-R_{2}\right) \left( S_{{\tau }
}-C_{\tau }\right) ^{+}+\left( 1-R_{3}\right) \left( S_{\tau }-C_{{\tau }
}\right) ^{-}\right) {\mathbb{I}}_{\left\{ t<{\tau =\tau }^{\left\{
4\right\} }\leq T\right\} } \\
& \ \ \ +B_{{\tau }}^{-1}\left( \delta _{{\tau }}^{1}-\left( 1-R_{2}\right) \left(
\delta _{{\tau }}^{1}-C_{\tau }\right) ^{+}\right) {\mathbb{I}}_{\left\{ t<{
\tau =\tau }^{\left\{ 5\right\} }\leq T\right\} } \\
& \ \ \ +B_{{\tau }}^{-1}\left(\delta _{{\tau }}^{1}+\left( 1-R_{3}\right) \left( \delta _{{\tau }}^{1}-C_{{\tau }}\right) ^{-}\right) {\mathbb{I}}_{\left\{ t<{\tau =\tau }^{\left\{6\right\} }\leq T\right\} } \\
& \ \ \ +B_{{\tau }}^{-1}\left( \delta _{{\tau }}^{1}-\left( 1-R_{2}\right) \left(
\delta _{{\tau }}^{1}-C_{\tau }\right) ^{+}+\left( 1-R_{3}\right) \left(
\delta _{{\tau }}^{1}-C_{{\tau }}\right) ^{-}\right) {\mathbb{I}}_{\left\{ t<
{\tau =\tau }^{\left\{ 7\right\} }\leq T\right\} } \\
& \ \ \ \left. \left. -\kappa \int_{\left] t,T\right] }B_{u}^{-1}{\mathbb{I}}
_{\left\{ {\tau }>u\right\} }du\right\vert \mathcal{G}_{t}\right).
\end{align*}
This proves the result.
\end{proof}

\noindent In case when $R_{2}=R_{3}=1$ process $S$ is the same as process $S^C$. Thus, we obtain from the above

\begin{corollary}
The ex-dividend price process $S$ of a counterparty risk-free CDS contract, can be represented as follows \footnote{We note that formula (\ref{StRep}) provides a representation of $S_t$, which is convenient for our purposes. The traditional representation of $S_t$, typically used in the context of counterparty risk free CDS contracts is $$S_{t}=B_{t}{\mathbb{E}}_{\mathbb{Q}}\left( \left. B_{{\tau_1 }}^{-1}{\mathbb{I}
}_{\left\{ t<{\tau_1 }\leq T\right\} }\delta _{{\tau_1 }}^{1}
-\kappa \int_{\left] t,T\right] \ \
}B_{u}^{-1}{\mathbb{I}}_{\left\{ {\tau_1 }>u\right\} }du\right\vert \mathcal{G}
_{t}\right)$$. },
\begin{equation}
S_{t}=B_{t}{\mathbb{E}}_{\mathbb{Q}}\left( \left. B_{{\tau }
}^{-1}\sum_{i=1}^{7}{\mathbb{I}}_{\left\{ t<{\tau =\tau }^{\left\{ i\right\}
}\leq T\right\} }\hat{\delta}_{{\tau }}^{i}-\kappa \int_{\left] t,T\right] \
\ }B_{u}^{-1}{\mathbb{I}}_{\left\{ {\tau }>u\right\} }du\right\vert \mathcal{
G}_{t}\right),  \label{pr1n}
\end{equation}
where $\hat{\delta}_{{t}}^{1} =\hat{\delta}_{{t}}^{5}=\hat{\delta}_{{t}}^{6}=\hat{
\delta}_{{t}}^{7}=\delta _{{t}}^{1}, \ \ \textrm{and} \ \ \hat{\delta}_{{t}}^{2} =\hat{\delta}_{{t}}^{3}=\hat{\delta}_{{t}}^{4}=S_{{t
}}.$
Thus,
\begin{align}
S_{t}&=B_{t}{\mathbb{E}}_{\mathbb{Q}}\left( \left. B_{{\tau }}^{-1}{\mathbb{I}
}_{\left\{ t<{\tau =\tau }_{1}\leq T\right\} }\delta _{{\tau }}^{1}+B_{{\tau
}}^{-1}\sum_{i=2}^{4}{\mathbb{I}}_{\left\{ t<{\tau =\tau }^{\left\{
i\right\} }\leq T\right\} }S_{{\tau }} \right. \right. \label{StRep} \\
& \ \ \ \left. \left. -\kappa \int_{\left] t,T\right]}B_{u}^{-1}{\mathbb{I}}_{\left\{ {\tau }>u\right\} }du\right\vert \mathcal{G}_{t}\right) . \notag
\end{align}
\end{corollary}

The following result is borrowed from \cite{Bielecki2008} (see Lemma 3.1 therein)

\begin{lemma}\label{kLemma}
The following equality holds ($\mathbb{Q}$ a.s.)
\begin{equation}\label{eq:kLeq}
B_{t}{\mathbb{E}}_{\mathbb{Q}}\left( \left. {\mathbb{I}}_{\left\{ t<{\tau
=\tau }^{\left\{ i\right\} }\leq T\right\} }B_{\tau }^{-1}\overline{\delta }
_{\tau }^{i}\right\vert \mathcal{G}_{t}\right) ={\mathbb{I}}_{\left\{ t<{
\tau }\right\} }\frac{B_{t}}{G(t) }{\mathbb{E}}_{\mathbb{Q}
}\left( \left. \int_{t}^{T}B_{u}^{-1}l_{u}^{i}\overline{\delta }
_{u}^{i}G\left( u\right) du\right\vert \mathcal{F}_{t}\right),
\end{equation}
for every $t\in \left[ 0,T\right]$.
\end{lemma}
The pre-default ex-dividend price processes, say $\widetilde{S}$ and $\widetilde{S}^{C}$, are defined as the (unique) $\bF$-adapted processes (cf. \cite{Bielecki2008}) such that
$$S_{t}^{C}={\mathbb{I}}_{\left\{t<\tau \right\} }\widetilde{S}_{t}^{C},\ \ S_{t}={\mathbb{I}}_{\left\{t<\tau \right\} }\widetilde{S}_{t} .$$
\noindent
In view of the above we thus obtain the following result
\begin{lemma} \label{nicerlemma}
We have that, for every $t\in [0,T]$,
\begin{equation}\label{Ife1}
\widetilde{S}_{t}^{C}=\frac{B_{t}}{G(t) }{\mathbb{E}}_{\mathbb{Q}
}\left( \left. \int_{t}^{T}B_{u}^{-1}G\left( u\right) \left(
\sum_{i=1}^{7}l_{u}^{i}\overline{\delta }_{u}^{i}-\kappa \right)
du\right\vert \mathcal{F}_{t}\right),
\end{equation}
and
\begin{equation}\label{Ife2}
\widetilde{S}_{t}=\frac{B_{t}}{G(t) }{\mathbb{E}}_{\mathbb{Q}
}\left( \left. \int_{t}^{T}B_{u}^{-1}G\left( u\right) \left(
\sum_{i=1}^{7}l_{u}^{i}\hat{\delta}_{u}^{i}-\kappa \right) du\right\vert
\mathcal{F}_{t}\right) .
\end{equation}

\end{lemma}

\begin{proof}
From Lemma \ref{ScRep} we have that
\begin{equation*}
S_{t}^{C}=B_{t}{\mathbb{E}}_{\mathbb{Q}}\left( \left. B_{{\tau }
}^{-1}\sum_{i=1}^{7}{\mathbb{I}}_{\left\{ t<{\tau =\tau }^{\left\{ i\right\}
}\leq T\right\} }\overline{\delta }_{t}^{i}\right\vert \mathcal{G}
_{t}\right) -\kappa B_{t}{\mathbb{E}}_{\mathbb{Q}}\left( \left. \int_{\left]
t,T\right] \ \ }B_{u}^{-1}{\mathbb{I}}_{\left\{ {\tau }>u\right\}
}du\right\vert \mathcal{G}_{t}\right) .
\end{equation*}
Now, in view of \eqref{eq:kLeq} we see that
$$B_{t}{\mathbb{E}}_{\mathbb{Q}}\left( \left. B_{{\tau }
}^{-1}\sum_{i=1}^{7}{\mathbb{I}}_{\left\{ t<{\tau =\tau }^{\left\{ i\right\}
}\leq T\right\} }\overline{\delta }_{t}^{i}\right\vert \mathcal{G}
_{t}\right)={\mathbb{I}}_{\left\{ t<{
\tau }\right\} }\frac{B_{t}}{G(t) }{\mathbb{E}}_{\mathbb{Q}}\left( \left. \sum_{i=1}^{7} \int_{t}^{T}B_{u}^{-1}l_{u}^{i}\overline{\delta }
_{u}^{i}G\left( u\right) du\right\vert \mathcal{F}_{t}\right). $$

Let us now fix $t\geq 0$, and define $Y_{s}:=-\kappa \int_{]t,s]}B_{u}^{-1}du$ for $s\geq t.$  Thus, we get
\begin{align*}
-\kappa B_{t}{\mathbb{E}}_{\mathbb{Q}}\left( \left. \int_{\left] t,T\right]
\ \ }B_{u}^{-1}{\mathbb{I}}_{\left\{ {\tau }>u\right\} }du\right\vert
\mathcal{G}_{t}\right) &=B_{t}{\mathbb{E}}_{\mathbb{Q}}\left( \left. {\mathbb{
I}}_{\left\{ {\tau }>T\right\} }Y_{T}\right\vert \mathcal{G}_{t}\right) \\
& \ \ \ +B_{t}{\mathbb{E}}_{\mathbb{Q}}\left( \left. {\mathbb{I}}_{\left\{ t<\tau
\leq T\right\} }Y_{\tau }\right\vert \mathcal{G}_{t}\right).
\end{align*}
It is known from \cite{Bielecki2008}, that
\begin{equation*}
B_{t}{\mathbb{E}}_{\mathbb{Q}}\left( \left. {\mathbb{I}}_{\left\{ t<\tau
\leq T\right\} }Y_{\tau }\right\vert \mathcal{G}_{t}\right) = -{\mathbb{I}}_{\left\{t<\tau
\right\} }\frac{B_{t}}{G(t)}{\mathbb{E}}_{\mathbb{Q}}\left( \left. \int_{t}^{T}Y_{u}dG(u) \right\vert \mathcal{F}_{t}\right)
\end{equation*}
and
\begin{equation*}
B_{t}{\mathbb{E}}_{\mathbb{Q}}\left( \left. {\mathbb{I}}_{\left\{ {\tau }>T\right\} }Y_{T}\right\vert \mathcal{G}_{t}\right) = {\mathbb{I}}_{\left\{t<\tau \right\} }\frac{B_{t}}{G(t)}{\mathbb{E}}_{\mathbb{Q}}\left( \left. G(T)Y_{T}\right\vert \mathcal{F}_{t}\right).
\end{equation*}
Finally, since $Y$ is of finite variation, (\ref{Ife1}) follows by applying the integration by parts formula
\begin{equation*}
G(t) Y_{T}-\int_{t}^{T}Y_{s}dG\left( s\right)=\int_{t}^{T}G\left( s\right) dY_{s}=-\kappa \int_{t}^{T}G\left( s\right)B_{u}^{-1}du.
\end{equation*}
Equality (\ref{Ife2}) is obtained as a special case of (\ref{Ife1}), by setting $R_2=R_3=1$.
\end{proof}

We are ready now to derive dynamics of the pre-default price processes, that we shall use in order to derive the dynamics of the CVA process.
\begin{lemma}\label{Inl}
(i) The pre-default ex-dividend price of a counterparty risky CDS contract follows the dynamics given as
\begin{align*}
d\widetilde{S}_{t}^{C} &=\left( \left( r_{t}+l_{t}\right) \widetilde{S}_{t}^{C}-\left( \sum_{i=1}^{7}l_{t}^{i}\overline{\delta }_{t}^{i}-\kappa\right) \right) dt+G^{-1}(t) \left( B_{t}dm_{t}^{C}-\widetilde{S}_{t}^{C}d\mu \right) \\
& \ \ \ +G^{-2}(t) \left( \widetilde{S}_{t}^{C}d\left\langle \mu\right\rangle _{t}-B_{t}d\left\langle \mu ,m^{C}\right\rangle _{t}\right), \ t\in [0,T],
\end{align*}
where
\begin{equation*}
m_{t}^{C}={\mathbb{E}}_{\mathbb{Q}}\left( \left.
\int_{0}^{T}B_{u}^{-1}G\left( u\right) \left( \sum_{i=1}^{7}l_{u}^{i}
\overline{\delta }_{u}^{i}-\kappa \right) du\right\vert \mathcal{F}
_{t}\right)
\end{equation*}
(ii) The pre-default ex-dividend price of a counterparty risk-free CDS contract follows the dynamics given as
\begin{align*}
d\widetilde{S}_{t} &= \left( \left( r_{t}+l_{t}\right) \widetilde{S}_{t}-\left( \sum_{i=1}^{7}l_{t}^{i}\hat{\delta}_{{t}}^{i}-\kappa \right)
\right) dt+G^{-1}(t) \left( B_{t}dm_{t}-\widetilde{S}_{t}d\mu\right) \\
& \ \ \ +G^{-2}(t) \left( \widetilde{S}_{t}d\left\langle \mu\right\rangle _{t}-B_{t}d\left\langle \mu ,m\right\rangle _{t}\right),\ t\in [0,T],
\end{align*}
where
\begin{equation*}
m_{t}={\mathbb{E}}_{\mathbb{Q}}\left( \left. \int_{0}^{T}B_{u}^{-1}G\left(
u\right) \left( \sum_{i=1}^{7}l_{u}^{i}\hat{\delta}_{{u}}^{i}-\kappa \right)
du\right\vert \mathcal{F}_{t}\right) .
\end{equation*}
\end{lemma}

\begin{proof} The argument below follows the one in the proof of Proposition 1.2 in \cite{Bielecki2008}.

In view of (\ref{Ife1}) we may write $\widetilde{S}_{t}^{C}$ as
\begin{equation*}
\widetilde{S}_{t}^{C}=B_{t}G^{-1}(t) U_{t},
\end{equation*}
where
\begin{equation*}
U_{t}=m_{t}^{C}-\int_{0}^{t}B_{u}^{-1}G\left( u\right) \left(\sum_{i=1}^{7}l_{u}^{i}\overline{\delta }_{u}^{i}-\kappa \right) du.
\end{equation*}
Since $G=\mu -v$, then applying It\^{o}'s formula one obtains
\begin{align*}
d\left( G^{-1}(t) U_{t}\right) &= G^{-1}(t)dm_{t}^{C}-B_{t}^{-1}\left( \sum_{i=1}^{7}l_{t}^{i}\overline{\delta }_{t}^{i}-\kappa \right) dt \\
& \ \ \ +U_{t}\left( G^{-3}(t) d\left\langle\mu \right\rangle _{t}-G^{-2}(t) \left( d\mu _{t}-dv_{t}\right)\right) \\
& \ \ \ -G^{-2}(t) d\left\langle \mu ,m^{C}\right\rangle _{t}.
\end{align*}
Consequently,

\begin{align*}
d\widetilde{S}_{t}^{C} &=B_{t}G^{-1}(t) dm_{t}^{C}-\left(\sum_{i=1}^{7}l_{t}^{i}\overline{\delta }_{t}^{i}-\kappa \right) dt & \ \ \ & \ \ \ \\ & \ \ \ +B_{t}U_{t}\left( G^{-3}(t) d\left\langle \mu \right\rangle_{t}-G^{-2}(t) \left( d\mu _{t}-l_{t}G(t) dt\right)\right) \\
& \ \ \ -B_{t}G^{-2}(t) d\left\langle \mu ,m^{C}\right\rangle+r_{t}B_{t}G^{-1}(t) U_{t}dt \\
&=\left( \left( r_{t}+l_{t}\right) \widetilde{S}_{t}^{C}-\left( \sum_{i=1}^{7}l_{t}^{i}\overline{\delta }_{t}^{i}-\kappa\right) \right)dt
+G^{-1}(t) \left( B_{t}dm_{t}^{C}-\widetilde{S}_{t}d\mu \right) \\
& \ \ \ +G^{-2}(t) \left( \widetilde{S}_{t}d\left\langle \mu\right\rangle _{t}-B_{t}d\left\langle \mu ,m^{C}\right\rangle _{t}\right),
\end{align*}
which verifies the result stated in (i).

Starting from (\ref{Ife2}), and using computations analogous to the ones done in (i), one can derive the result stated in (ii).
\end{proof}

The dynamics of the CVA process are easily derived with help of the above lemma,

\begin{proposition}\label{Inp}
The bilateral CVA process satisfies,
\begin{align*}
d\text{CVA}_{t} &=r_{t}\text{CVA}_{t}dt-\text{CVA}_{t-}dM_{t}-\left(1-H_{t}\right) \left( \sum_{i=1}^{7}l_{t}^{i}\xi _{t}^{i}\right) dt \\
& \ \ \ +\left(1-H_{t}\right) B_{t}G^{-1}(t) dn_{t}-G^{-1}(t) \text{CVA}_{t}d\mu _{t}+G^{-2}(t) \text{CVA}_{t}d\left\langle \mu \right\rangle _{t} \\
& \ \ \ -\left( 1-H_{t}\right)G^{-2}(t) B_{t}\left( d\left\langle \mu ,m\right\rangle_{t}-d\left\langle \mu ,m^{C}\right\rangle _{t}\right) ,
\end{align*}
where
\begin{equation*}
n_{t}={\mathbb{E}}_{\mathbb{Q}}\left( \left. \int_{0}^{T}B_{u}^{-1}G\left(
u\right) \left( \sum_{i=1}^{7}l_{u}^{i}\xi _{u}^{i}\right) du\right\vert
\mathcal{F}_{t}\right) ,\ t\in [ 0,T] ,
\end{equation*}
with
\begin{eqnarray*}
&&\xi _{t}^{1}=0,\text{ }\xi _{t}^{2}=\left( 1-R_{2}\right) \left( S_{{t}
}-C_{{t}}\right) ^{+},\text{ }\xi _{t}^{3}=-\left( 1-R_{3}\right) \left( S_{{
t}}-C_{{t}}\right) ^{-}, \\
&&\xi _{t}^{4}=\left( 1-R_{2}\right) \left( S_{{t}}-C_{{t}}\right) ^{+}\text{
}-\left( 1-R_{3}\right) \left( S_{{t}}-C_{{t}}\right) ^{-}, \\
&&\xi _{t}^{5}=\left( 1-R_{2}\right) \left( \delta _{{t}}^{1}-C_{{t}}\right)
^{+},\text{ }\xi _{t}^{6}=-\left( 1-R_{3}\right) \left( \delta _{{t}}^{1}-C_{
{t}}\right) ^{-}, \\
&&\xi _{t}^{7}=\left( 1-R_{2}\right) \left( \delta _{{t}}^{1}-C_{t}\right)
^{+}-\left( 1-R_{3}\right) \left( \delta _{{t}}^{1}-C_{{t}}\right) ^{-}.
\end{eqnarray*}

\end{proposition}

\begin{proof}
Applying the integration by parts formula we get that
\begin{equation*}
d\text{CVA}_{t}=\left( 1-H_{t}\right) \left( d\widetilde{S}_{t}-d\widetilde{S
}_{t}^{C}\right) -\left( \widetilde{S}_{t}-\widetilde{S}_{t}^{C}\right)
dH_{t}.
\end{equation*}
This together with Lemma \ref{Inl} implies
\begin{align*}
d\text{CVA}_{t} &=-\left( S_{t-}-S_{t-}^{C}\right) dM_{t}+\left(1-H_{t}\right) \left( r_{t}\left( S_{t}-S_{t}^{C}\right)-\sum_{i=1}^{7}l_{t}^{i}\left( \hat{\delta}_{{t}}^{i}-\overline{\delta }_{t}^{i}\right) \right) dt \\
& \ \ \ +\left( 1-H_{t}\right)B_{t}G^{-1}(t)\left(dm_{t}-dm_{t}^{C}\right)-\left( 1-H_{t}\right)G^{-1}(t)\left(S_{t}-S_{t}^{C}\right)d\mu _{t} \\
& \ \ \ +\left( 1-H_{t}\right) G^{-2}(t) \left( S_{t}-S_{t}^{C}\right)d\left\langle \mu \right\rangle _{t} \\
& \ \ \ -\left( 1-H_{t}\right) G^{-2}(t) B_{t}\left( d\left\langle \mu ,m\right\rangle _{t}-d\left\langle\mu ,m^{C}\right\rangle _{t}\right),
\end{align*}
which proves the result.
\end{proof}

\subsubsection{Dynamics of CVA when the immersion property holds}

 Here we adapt the results derived above to the case when the immersion property holds between filtrations $\bF$ and $\bG$, that is the case when every $\bF$-martingale is a $\bG$-martingale under $\mathbb{Q}$. In this case, the continuous martingale $\mu $ in the Doob-Meyer decomposition of $G$ vanishes, so that the survival process $G$ is a non-increasing process represented as $G=-v.$  Frequently, the immersion property is referred to as Hypothesis $\left(\mathcal{H}\right) $. For an excellent discussion of the immersion property we refer to \cite{JEANBLANC2009}.

\begin{assumption}\label{a1}
Hypothesis $\left(\mathcal{H}\right) $ holds between the filtrations $\mathbb{F}$ and $\mathbb{G}$ under $\mathbb{Q}$.
\end{assumption}

In view of the results (and the notation) from Proposition \ref{Inp} we obtain
\begin{corollary}
Assume that Assumption \ref{a1} is satisfied. Then,
\begin{align*}
d\text{CVA}_{t} &=r_{t}\text{CVA}_{t}dt-\text{CVA}_{t-}dM_{t}-\left(1-H_{t}\right) \left( \sum_{i=1}^{7}l_{t}^{i}\xi _{t}^{i}\right) dt \\
& \ \ \ +\left(1-H_{t}\right) B_{t}G^{-1}(t) dn_{t}, \ t\in [0,T].
\end{align*}

\end{corollary}

\begin{remark}
If we assume that the filtration $\mathbb{F}$ is generated by a Brownian motion, then, in view of the Brownian martingale representation theorem,  there exists an $\mathbb{F}$-predictable process $\zeta $\ such that $dn_{t}=\zeta _{t}dW_{t}$ and
\begin{align*}
d\text{CVA}_{t}&=r_{t}\text{CVA}_{t}dt-\text{CVA}_{t-}dM_{t}-\left(1-H_{t}\right) \left( \sum_{i=1}^{7}l_{t}^{i}\xi _{t}^{i}\right) dt \\
& \ \ \ +\left(1-H_{t}\right) B_{t}G^{-1}(t) \zeta _{t}dW_{t}.
\end{align*}
\end{remark}

\subsection{Fair Spread Value Adjustment}\label{sec:SVA}

Let us fix $t\in [0,T],$ and let us denote by $\kappa _{t}$ the market spread of the counterparty risk-free CDS contract at time $t$; that is, $\kappa _{t}$ is this level of spread  that makes the pre-default values of the two legs of a counterparty risk-free CDS contract equal to each other at time $t$,
\begin{equation}
\widetilde S_{{t}}\left( \kappa _{t}\right)=0. \label{sp1}
\end{equation}
It is convenient to write the above equation in the form that is common in practice:
\begin{equation}
 PL_{t}-\kappa _{t}RDV01_{t}=0  \label{sp11},
\end{equation}
 where $PL$ and $RDV01$ are processes representing (pre-default) values of the protection leg and the risky annuity, respectively, so that$\, $\footnote{We note that formula (\ref{PL1}) provides a representation of $PL_t$, which is convenient for our purposes. The traditional representation of $PL_t$, typically used in the context of counterparty risk free CDS contracts is $$PL_{t}=\frac{B_{t}}{G^{1}\left(
t\right) }{\mathbb{E}}_{\mathbb{Q}}\left( \left. \int_{\left] t,T\right] }B_{
{u}}^{-1}G^{1}\left( u\right) \delta _{{u}}^{1}\lambda^1_u du\right\vert \mathcal{F}_{{t}}\right),$$ where $\lambda^1$ is the $\bF$ intensity of $\tau_1$. }

\begin{equation}\label{PL1}
PL_{t}=\frac{B_{t}}{G^{1}\left(
t\right) }{\mathbb{E}}_{\mathbb{Q}}\left( \left. \int_{\left] t,T\right] }B_{
{u}}^{-1}G^{1}\left( u\right) \delta _{{u}}^{1}\left(
\sum_{i=1,5,6,7}l_{u}^{i}\right) du\right\vert \mathcal{F}_{{t}}\right) ,
\end{equation}
and
\begin{equation}\label{PL2}
RDV01_{t}=\frac{B_{t}}{
G^{1}(t) }{\mathbb{E}}_{\mathbb{Q}}\left( \left. \int_{]
t,T] }B_{u}^{-1}G^{1}\left( u\right) du\right\vert
\mathcal{F}_{{t}}\right),
\end{equation}
where
\begin{equation*}
G^{1}(t) =\mathbb{Q}\left( \left. \tau _{1}>t\right\vert
\mathcal{F}_{t}\right) .
\end{equation*}

Therefore, we get,
\begin{equation}
\kappa _{t}=\frac{{\mathbb{E}}_{\mathbb{Q}}\left( \left. \int_{\left] t,T
\right] }B_{{u}}^{-1}G^{1}\left( u\right) \delta _{{u}}^{1}\left(
\sum_{i=1,5,6,7}l_{u}^{i}\right) du\right\vert \mathcal{F}_{{t}}\right) }{{
\mathbb{E}}_{\mathbb{Q}}\left( \left. \int_{\left] t,T\wedge {\tau }_{{1}}
\right] }B_{u}^{-1}G^{1}\left( u\right) du\right\vert \mathcal{F}_{{t}
}\right) }.  \label{kt}
\end{equation}

We denote by $\kappa _{t}^{C}$ the spread which makes the values of the two pre-first-default legs of a counterparty risky CDS contract equal to each other at every $t\in
\left[ 0,T\right] $ as
\begin{equation}
\widetilde S_{{t}}^{C}\left( \kappa _{t}^{C}\right) =PL_{t}^{C}-\kappa
_{t}^{C}RDV01_{t}^{C}=0.  \label{sp2}
\end{equation}

Similarly, we use the spread $\kappa _{0}^{C}$ initiated at time $t=0$ in order to compute the fair price of a counterparty risky CDS contract at any time $t\in \left[ 0,T\right] $. Using Lemma 3.1, $\kappa _{t}^{C}$ admits
the following representation for every $t\in \left[ 0,T\right] ,$
\begin{equation*}
\kappa _{t}^{C}=\frac{PL_{t}^{C}}{RDV01_{t}^{C}},
\end{equation*}
where

\begin{equation}
PL_{t}^{C}=\frac{B_{t}}{G\left(
t\right) }{\mathbb{E}}_{\mathbb{Q}}\left( \left.
\int_{t}^{T}B_{u}^{-1}G\left( u\right) \left( \sum_{i=1}^{7}l_{u}^{i}
\overline{\delta }_{u}^{i}\right) du\right\vert \mathcal{F}_{t}\right)
\end{equation}
and
\begin{equation}
RDV01_{t}^{C}=\frac{B_{t}}{G\left(
t\right) }{\mathbb{E}}_{\mathbb{Q}}\left( \left. \int_{] t,T] }B_{u}^{-1}G\left( u\right) du\right\vert \mathcal{F}_{{t}
}\right) .  \label{FLc}
\end{equation}

We may now introduce the following definition,

\begin{definition}
The Spread Value Adjustment process of a counterparty risky CDS contract maturing at time $T$ is defined as,
\begin{equation*}
\text{SVA}_{t}=\kappa _{t}-\kappa _{t}^{C}
\end{equation*}
for every $t\in \left[ 0,T\right] .$
\end{definition}

Monitoring SVA is of great importance since it provides a more practical way to quantify the counterparty risk. Moreover, the spread difference is a very useful indicator for the trading decisions in practice (cf. \cite{Gregory2009}).

\begin{proposition}
The SVA of a counterparty risky CDS contract maturing at time $T$ equals,
\begin{equation*}
\text{SVA}_{t}=\frac{\widetilde{\text{CVA}}_{t}}{B_{t}G^{-1}(t) {\mathbb{E}}
_{\mathbb{Q}}\left( \left. \int_{] t,T]
}B_{u}^{-1}G\left( u\right) du\right\vert \mathcal{F}_{{t}}\right) }
\end{equation*}
for every $t\in \left[ 0,T\right] ,$ where the pre-first-default bilateral Credit Valuation Adjustment process $\widetilde{\text{CVA}}$ is given as
\begin{equation}
\widetilde{\text{CVA}}_{t}=\widetilde S_{t}-\widetilde S_{t}^{C},
\label{CVAdef1}
\end{equation}
for every $t\in [0,T]$.
\end{proposition}

\begin{proof}
Let us rewrite $PL^{C}$ as
\begin{equation*}
PL_{t}^{C}=PL_{t}^{C}-\kappa _{t}RDV01_{t}^{C} +\kappa _{t}RDV01_{t}^{C}
\end{equation*}
by a simple modification. Now, using $\left( \ref{CVAdef}\right) $ and $\left( \ref{sp2}\right) $, we conclude that
\begin{align*}
PL_{t}^{C} &=\widetilde S_{t}^{C}(\kappa _{t}) + \kappa _{t}RDV01_{t}^{C} \\
&=\widetilde S_{t}(\kappa _{t}) - \widetilde{\text{CVA}}_{t} + \kappa _{t}RDV01_{t}^{C} .
\end{align*}
Since $\widetilde S_{t}(\kappa _{t}) =0$, then $\kappa _{t}^{C}$ has the following form,
\begin{equation*}
\kappa _{t}^{C}=\frac{-\widetilde{\text{CVA}}_{t}+\kappa _{t}RDV01_{t}^{C} }{RDV01_{t}^{C} },
\end{equation*}
which is
\begin{equation*}
\kappa _{t}^{C}=-\frac{\widetilde{\text{CVA}}_{t}}{RDV01_{t}^{C}}+\kappa _{t}.
\end{equation*}
\end{proof}

\subsubsection{SVA Dynamics}

Applying It\^{o} formula one obtains the dynamics of the fair spread process and of the counterparty risk adjusted spread process as

\begin{align}
d\kappa _{t} &=\frac{1}{\widetilde{RDV01}_{t}}\left( B_{t}^{-1}G^{1}\left(
t\right) \left( \kappa _{t}-\delta _{{t}}^{1}l_{t}^{1}\right) dt+\frac{
\kappa _{t}}{\widetilde{RDV01}_{t}}d\left\langle \eta ^{2}\right\rangle _{t} \right. \label{dkt} \\
& \ \ \ \left. -\frac{1}{\widetilde{RDV01}_{t}}d\left\langle \eta ^{1},\eta ^{2}\right\rangle
_{t}\right) +\frac{1}{\widetilde{RDV01}_{t}}\left( d\eta _{t}^{1}-\kappa _{t}d\eta
_{t}^{2}\right) , \qquad \qquad\qquad t\in[0,T], \notag
\end{align}
where
\begin{equation*}
\widetilde{RDV01}_{t} :={\mathbb{E}}_{\mathbb{Q}}\left( \left. \int_{]
t,T] }B_{u}^{-1}G^{1}\left( u\right) du\right\vert
\mathcal{F}_{{t}}\right) ,
\end{equation*}

\begin{equation*}
\eta _{t}^{1} :={\mathbb{E}}_{\mathbb{Q}}\left( \left. \int_{\left] 0,T\right]
}B_{{u}}^{-1}G^{1}\left( u\right) \delta _{{u}}^{1}l_{u}^{1}du\right\vert
\mathcal{F}_{{t}}\right) ,
\end{equation*}
\begin{equation*}
\eta _{t}^{2}={\mathbb{E}}_{\mathbb{Q}}\left( \left. \int_{\left] 0,T\right]
}B_{u}^{-1}G^{1}\left( u\right) du\right\vert \mathcal{F}_{{t}}\right) =
\widetilde{RDV01}_{t}+\int_{\left] 0,t\right] }B_{u}^{-1}G^{1}\left( u\right)
du,
\end{equation*}
and
\begin{align}
d\kappa _{t}^{C} &=\frac{1}{\widetilde{RDV01}_{t}^{C}}\left(
B_{t}^{-1}G(t) \left( \kappa _{t}^{C}-\sum_{i=1}^{7}\widetilde{
\delta }_{{t}}^{i}l_{t}^{i}\right) dt+\frac{\kappa _{t}^{C}}{\widetilde{RDV01}
_{t}^{C}}d\left\langle \zeta ^{2}\right\rangle _{t} \right. \label{dkct} \\
& \ \ \ \left. -\frac{1}{\widetilde{RDV01}_{t}^{C}}d\left\langle \zeta ^{1},\zeta ^{2}\right\rangle _{t}\right)+\frac{1}{\widetilde{RDV01}_{t}^{C}}\left( d\zeta _{t}^{1}-\kappa
_{t}^{C}d\zeta _{t}^{2}\right) ,  \notag
\end{align}
where
\begin{equation*}
\widetilde{RDV01}_{t}={\mathbb{E}}_{\mathbb{Q}}\left( \left. \int_{]
t,T] }B_{u}^{-1}G^{1}\left( u\right) du\right\vert
\mathcal{F}_{{t}}\right) ,
\end{equation*}
with
\begin{equation*}
\zeta _{t}^{1}={\mathbb{E}}_{\mathbb{Q}}\left( \left. \int_{\left] 0,T\right]
}B_{{u}}^{-1}G\left( u\right) \left( \sum_{i=1}^{7}l_{u}^{i}\overline{\delta
}_{u}^{i}\right) du\right\vert \mathcal{F}_{{t}}\right),
\end{equation*}
and
\begin{equation*}
\zeta _{t}^{2}={\mathbb{E}}_{\mathbb{Q}}\left( \left. \int_{\left] 0,T\right]
}B_{u}^{-1}G\left( u\right) du\right\vert \mathcal{F}_{{t}}\right) =%
\widetilde{RDV01}_{t}^{C}+\int_{\left] 0,t\right] }B_{u}^{-1}G\left( u\right)
du.
\end{equation*}

Combining the above results, we find the dynamics of the SVA process:
\begin{equation*}
d\, \textrm{SVA}_{t}=d\kappa _{t}-d\kappa _{t}^{C}, \quad t\in[0,T].
\end{equation*}

Dynamics of the SVA is of great importance for observing the behavior of the difference between the fair spread and the counterparty risk adjusted spread. Counterparty risk dynamics can be assessed in a more intuitive manner by computing the SVA dynamics.

\section{Multivariate Markovian Default Model}\label{sec:Numerics}
In this section, we propose an underlying stochastic model following the lines of \cite{Bielecki2011}. Towards this end we define a Markovian model of multivariate default times with factor processes $X=\left(X^{1},X^{2},X^{3}\right) $ which will have the following key features,

\begin{itemize}
\item The pair $\left( X,H\right) $ is Markov in its natural filtration,

\item Each pair $\left( X^{i},H^{i}\right) $ is a Markov process,

\item At every instant, either each counterparty defaults individually or simultaneously with other counterparties.
\end{itemize}

Note that the second property grants quick valuation of the CDS and independent calibration of each model marginal $\left( X^{i},H^{i}\right)$, whereas the third property will allow us to account for dependence between defaults. We present here some numerical results as an application of above theory. The default intensities are assumed to be of the affine form
\begin{equation*}
l_{i}\left( t,X_{t}^{i}\right) =a_{i}+X_{t}^{i},
\end{equation*}
where $a_{i}$ is a constant and $X^{i}$ is a homogenous CIR process generated by,
\begin{equation*}
dX_{t}^{i}=\zeta _{i}\left( \mu _{i}-X_{t}^{i}\right) dt-\sigma _{i}\sqrt{
X_{t}^{i}}dW_{t}^{i},
\end{equation*}
for $i=1,2,3$. Each collection of the parameters $\left( \zeta _{i},\mu_{i},\sigma _{i}\right) $ may take values corresponding to a low, a medium or a high regime which are given as follows.

\begin{table}[h] \centering
{\begin{tabular}{|l||l|l|l|l|} \hline
Credit Risk Level & $\zeta $ & $\mu $ & $\sigma $ & $X_{0}$ \\ \hline
Low & 0.9 & 0.001 & 0.01 & 0.001 \\ \hline
Medium & 0.8 & 0.02 & 0.1 & 0.02 \\ \hline
High & 0.5 & 0.05 & 0.2 & 0.05 \\ \hline
\end{tabular}
}
\label{Table0}

\end{table}

Moreover, following the methodology in \cite{Bielecki2011}, we specify the marginal default intensity processes as follows
\begin{equation*}
q_{t}^{1}=l_{t}^{1}+l_{t}^{5}+l_{t}^{6}+l_{t}^{7},\text{ }
q_{t}^{2}=l_{t}^{2}+l_{t}^{4}+l_{t}^{5}+l_{t}^{7},\text{ }
q_{t}^{3}=l_{t}^{3}+l_{t}^{4}+l_{t}^{6}+l_{t}^{7}
\end{equation*}
where the related survival probabilities are found as
\begin{equation*}
\mathbb{Q}\left( {\tau }_{{i}}>t\right) ={\mathbb{E}}_{\mathbb{Q}}\left(e^{-\int_{0}^{t}q_{u}^{i}du}\right) \text
{ and }
\mathbb{Q}\left( {\tau }>t\right) ={\mathbb{E}}_{\mathbb{Q}}\left( e^{-\int_{0}^{t}l_{u}du}\right) .
\end{equation*}
For a detailed discussion including implementation and the calibration of the model, we refer to \cite{Bielecki2011} and \cite{Assefa2011}.

\subsection{Results}

 Our aim here is to assess by means of numerical experiments the impact of collateralization on the counterparty risk exposure. We present numerical results for different collateralization regimes distinguished by different threshold values. The numerical experiments below have been done using the three factor (2F)\ parametrization given in \cite{Bielecki2011}, the recovery rates are fixed to $40\%$, the risk-free rate $r$ is taken as $0$ and the maturity is set to $T=5$ years.

 Table \ref{Table1} shows the values of CVA$_0$ and SVA$_0$ for different threshold regimes. Threshold values are chosen as a fraction of the notional (cf. \cite{Pykhtin2009}). Computations are done assuming that (refer to Table \ref{Table0}) the underlying entity, the counterparty, and the investor has high risk levels. Simulated fair spread without counterparty risk is found as 153bps.  Case A represents the uncollateralized regime where there is no collateral exchanged (this is done by setting the thresholds infinity), whereas other Case F corresponds to the full collateralization where the thresholds are set to $0$. In each case, computations are done by setting $MTA$ to zero and assuming there is no margin period. One can observe that decreasing threshold value dramatically decreases the initial CVA and therefore the SVA values.

\begin{table}[h] \centering
{\begin{tabular}{l|l|l|l|l|}
\cline{2-5}
 & $\Gamma _{cpty}$ & $\Gamma _{inv}$ & CVA$_{0}$ & SVA$_{0}$\\ \cline{2-5}
Case A & $\infty $ & -$\infty $ & $1.01\times 10^{-4}$ & $0.2153$  \\ \cline{2-5}
Case B & $1.5\times 10^{-3}$ & $0.4\times 10^{-3}$ & $6.13\times 10^{-5}$ & $0.1305$\\ \cline{2-5}
Case C & $1\times 10^{-3}$ & $0.2\times 10^{-3} $& $4.36\times 10^{-5}$ & $0.0931$ \\ \cline{2-5}
Case D & $0.5\times 10^{-3}$ & $0.1\times 10^{-3}$ & $2.18\times 10^{-5}$ & $0.0464$\\ \cline{2-5}
Case E & $0.25\times 10^{-3}$ & $0.05\times 10^{-3}$ & $1.14\times 10^{-5}$ & $0.0243$\\ \cline{2-5}
Case F & $0$ & $0$ & $0$ & $0$ \\ \cline{2-5}
\end{tabular}
}

\label{Table1}
\end{table}

In Figure \ref{fig1}, we present the $EPE$ and $ENE$ curves for each case A to F, and we also plot the mean collateral values. Computations are carried out by running 10$^{4}$ Monte Carlo simulations. It is apparent that the behavior of the $EPE$ and $ENE$ values decreases as a result of increased collateralization. Note that there are peaks in the collateral value in the very beginning and through the maturity. This effect can be explained as follows: Observe from Table 1 that the investor has lower threshold than the counterparty in each cases from A to F. As a result, having a lower threshold value, investor will be posting collateral before the counterparty. Therefore, until the counterparty's exposure reaches the threshold, the collateral value remains negative; meaning that there will be margin calls for the investor before the counterparty.

Figure \ref{fig2} plots the mean of sample CVA paths. Starting from CVA$_0$ we compute the mean sample paths in each case. The behavior of CVA as a credit hybrid option, as indicated in Remark \ref{CVAoption}, can be clearly observed in the graphs. CVA values decrease over time as a result of time decay since the expected loss decreases close to the expiration. The effect of collateralization on the CVA values is apparent in the graphs. Observe that increased initial threshold values are of great importance since one can significantly reduce the future CVA values by appropriately setting the collateral thresholds. Moreover, one can also use dynamic thresholds by linking the threshold values to the counterparties' default intensities or credit ratings. In this way, counterparties will have more control on the future values of the CVA of the CDS contract and dynamically manage the CVA since the collateral thresholds will be reacting to the changes in the default intensities or credit ratings. This approach will be further investigated in a future research.

\begin{figure}[h]
\centering
\begin{tabular}{cc}
\epsfig{file=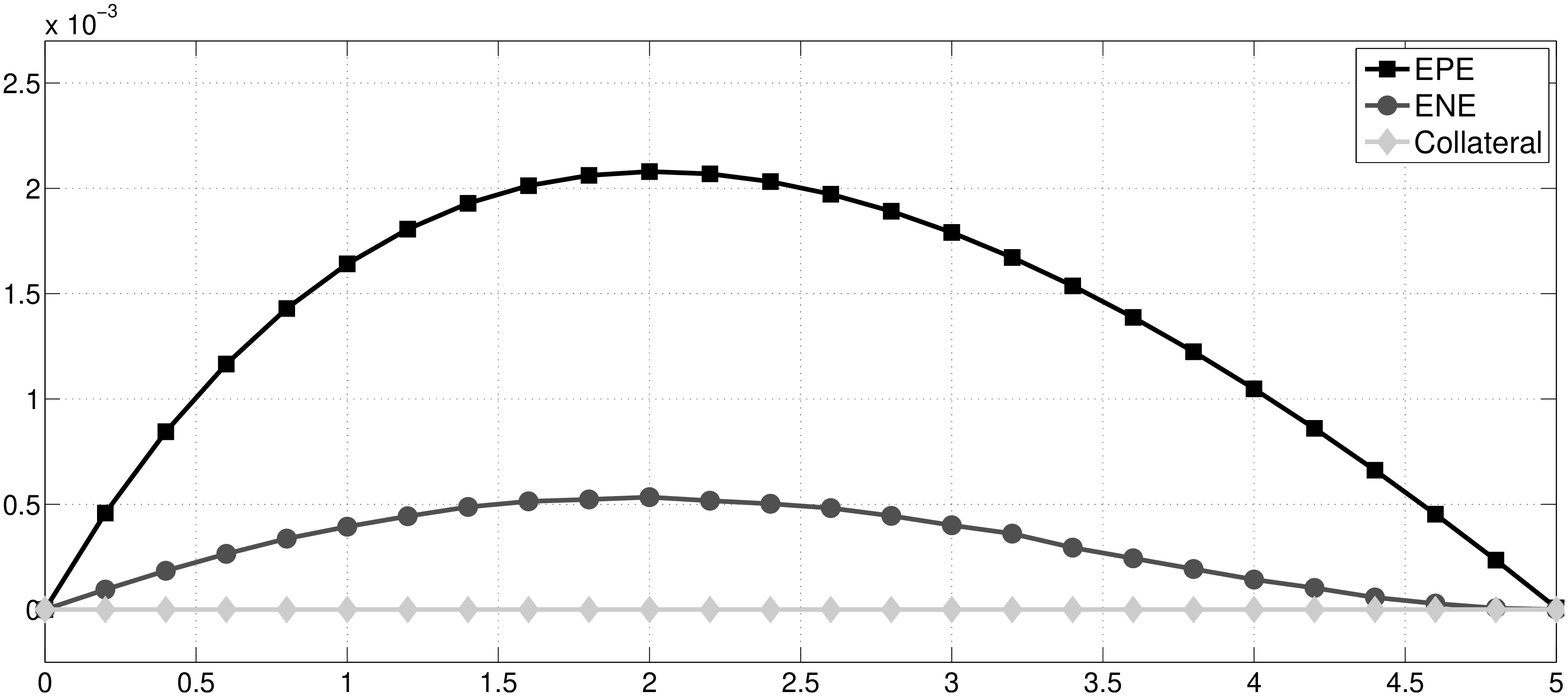,width=0.46\linewidth,clip=} &
\epsfig{file=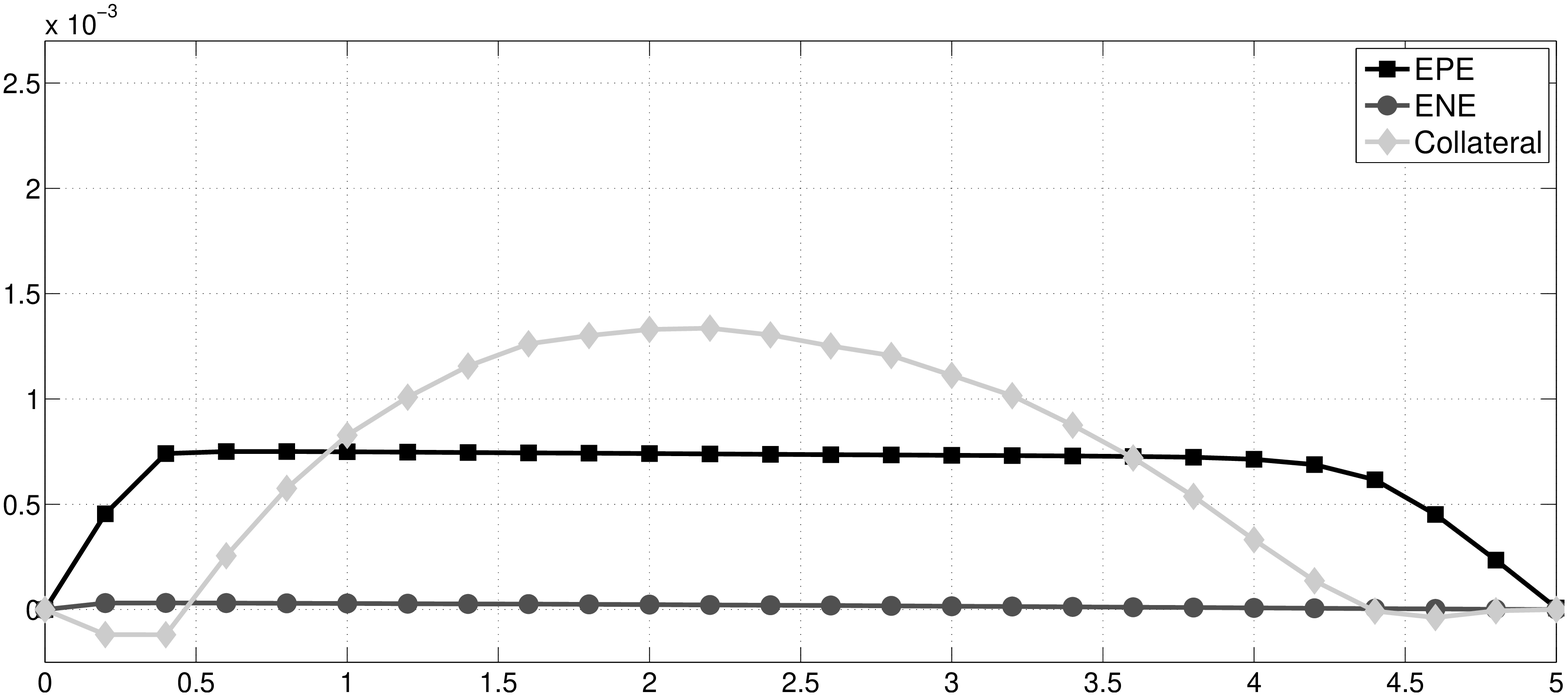,width=0.46\linewidth,clip=} \\
\epsfig{file=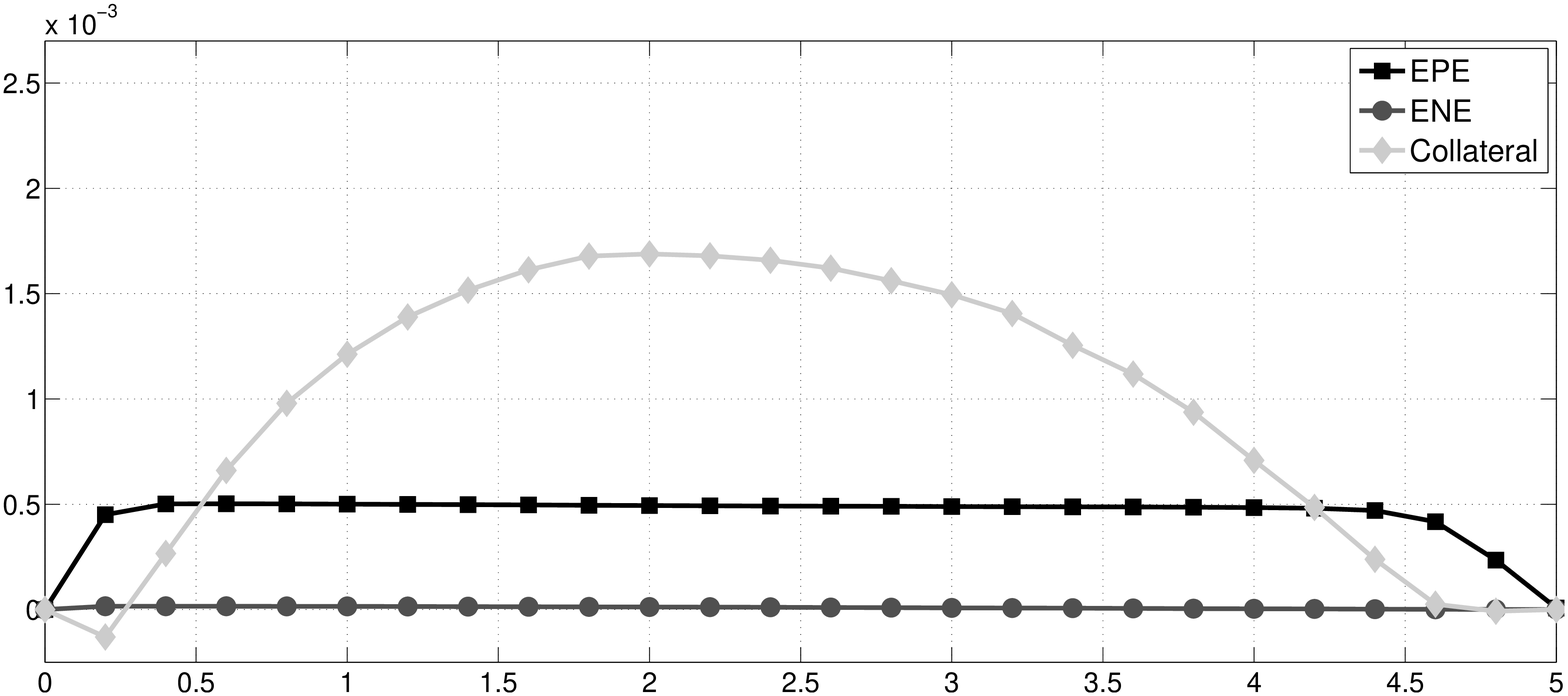,width=0.46\linewidth,clip=} &
\epsfig{file=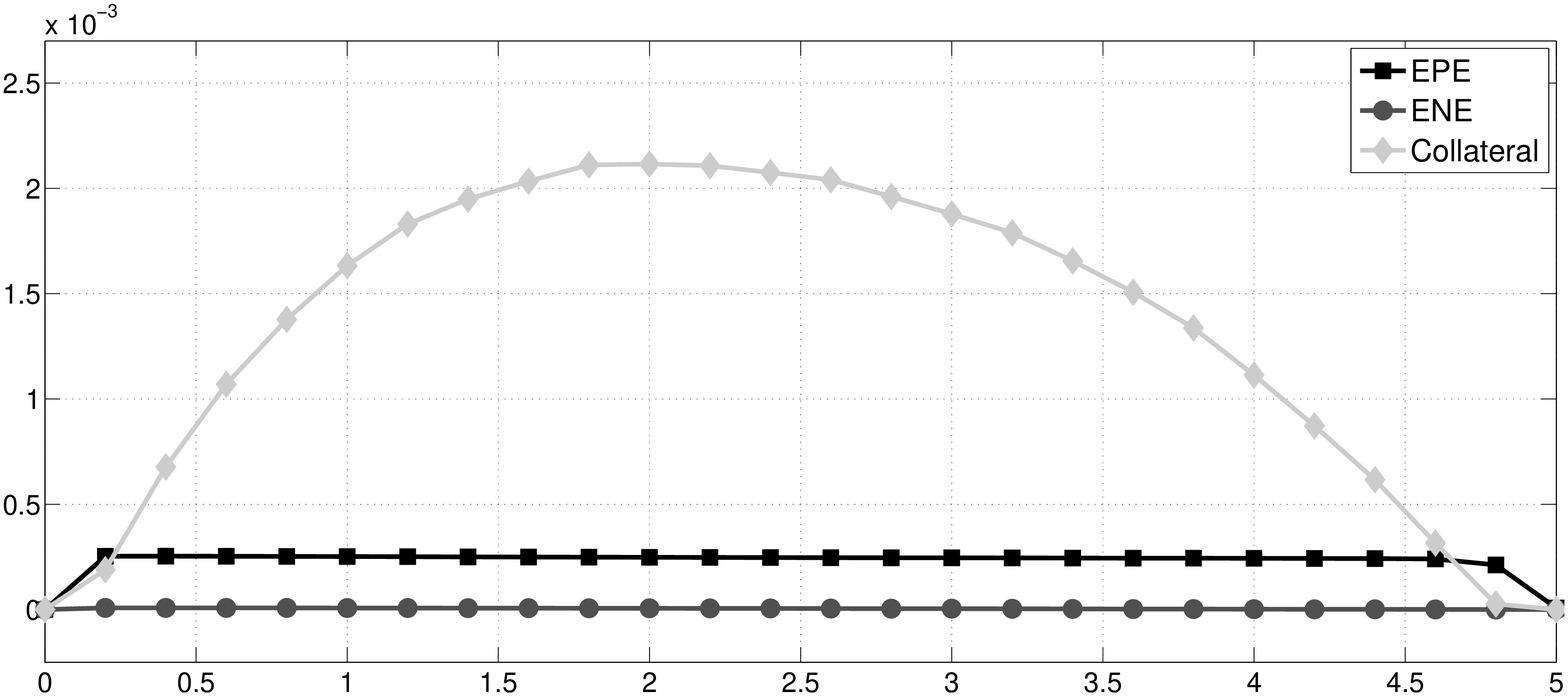,width=0.46\linewidth,clip=} \\
\epsfig{file=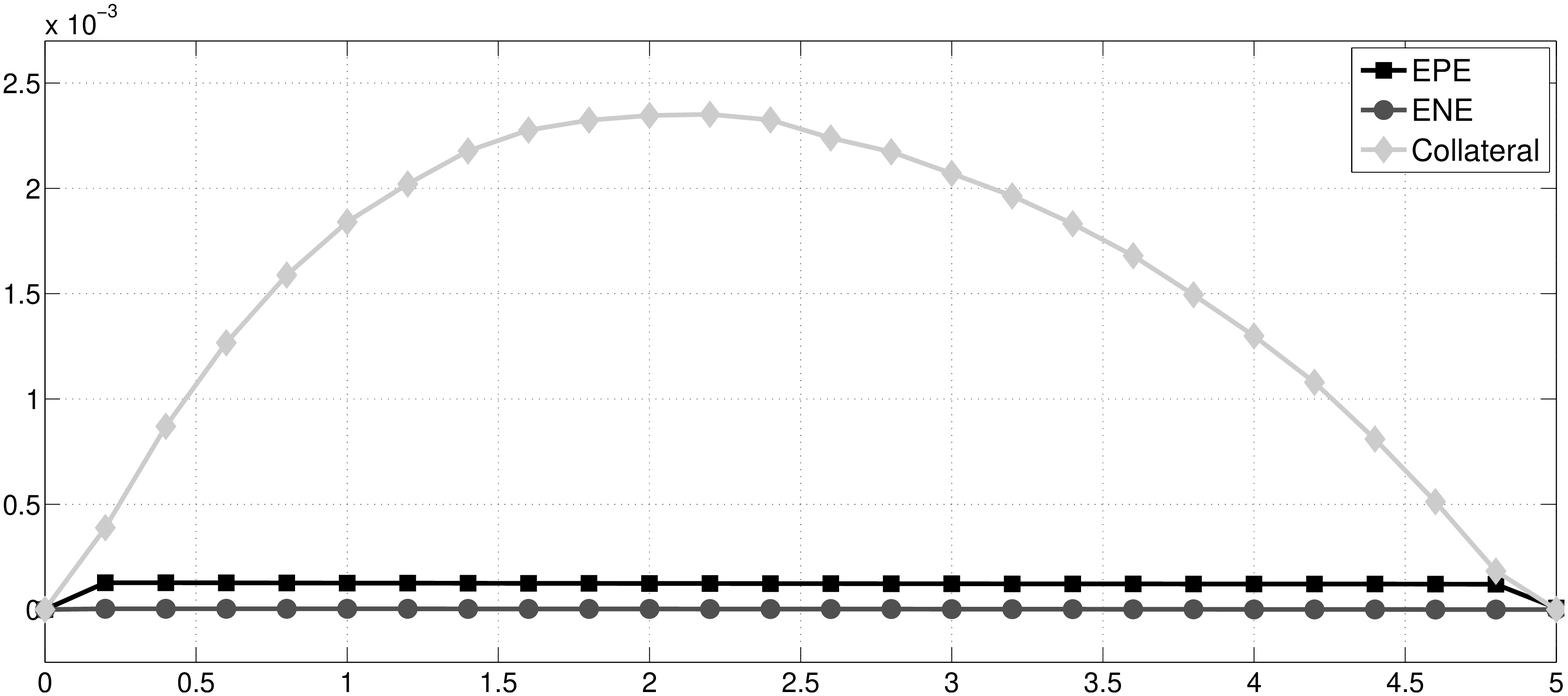,width=0.46\linewidth,clip=} &
\epsfig{file=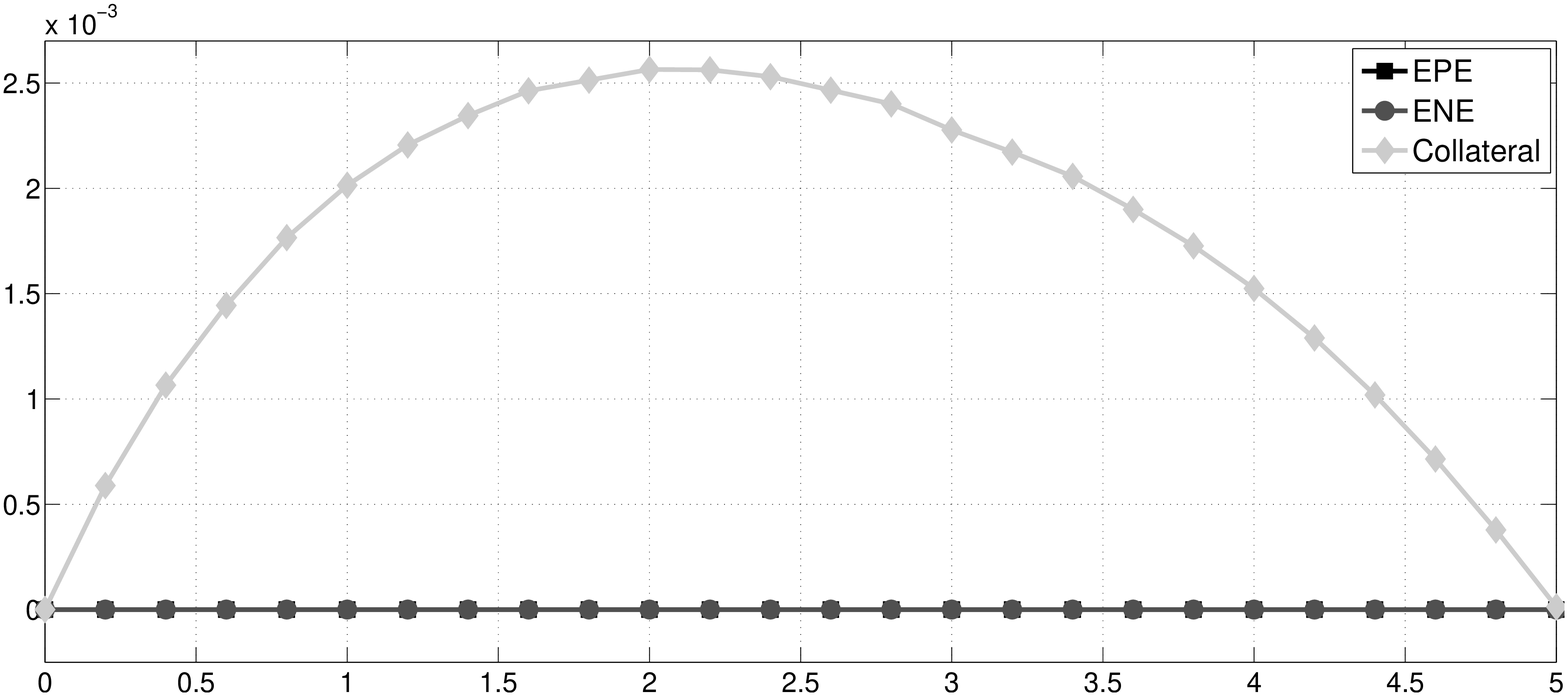,width=0.46\linewidth,clip=}
\end{tabular}
\caption{EPE, ENE and the Collateral curves for Case A, B, C, D, E, and F}
\label{fig1}
\end{figure}

\begin{figure}[h]
\centering
\begin{tabular}{cc}
\epsfig{file=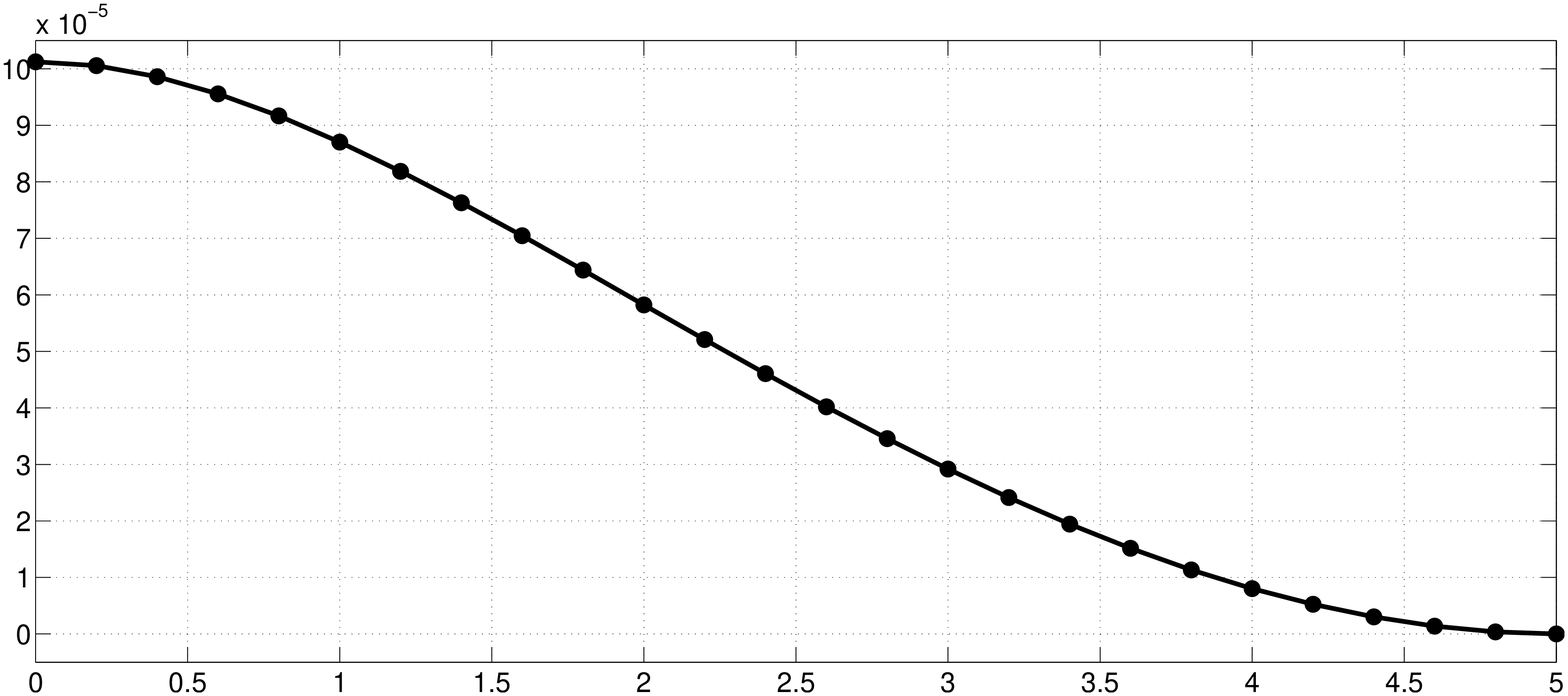,width=0.46\linewidth,clip=} &
\epsfig{file=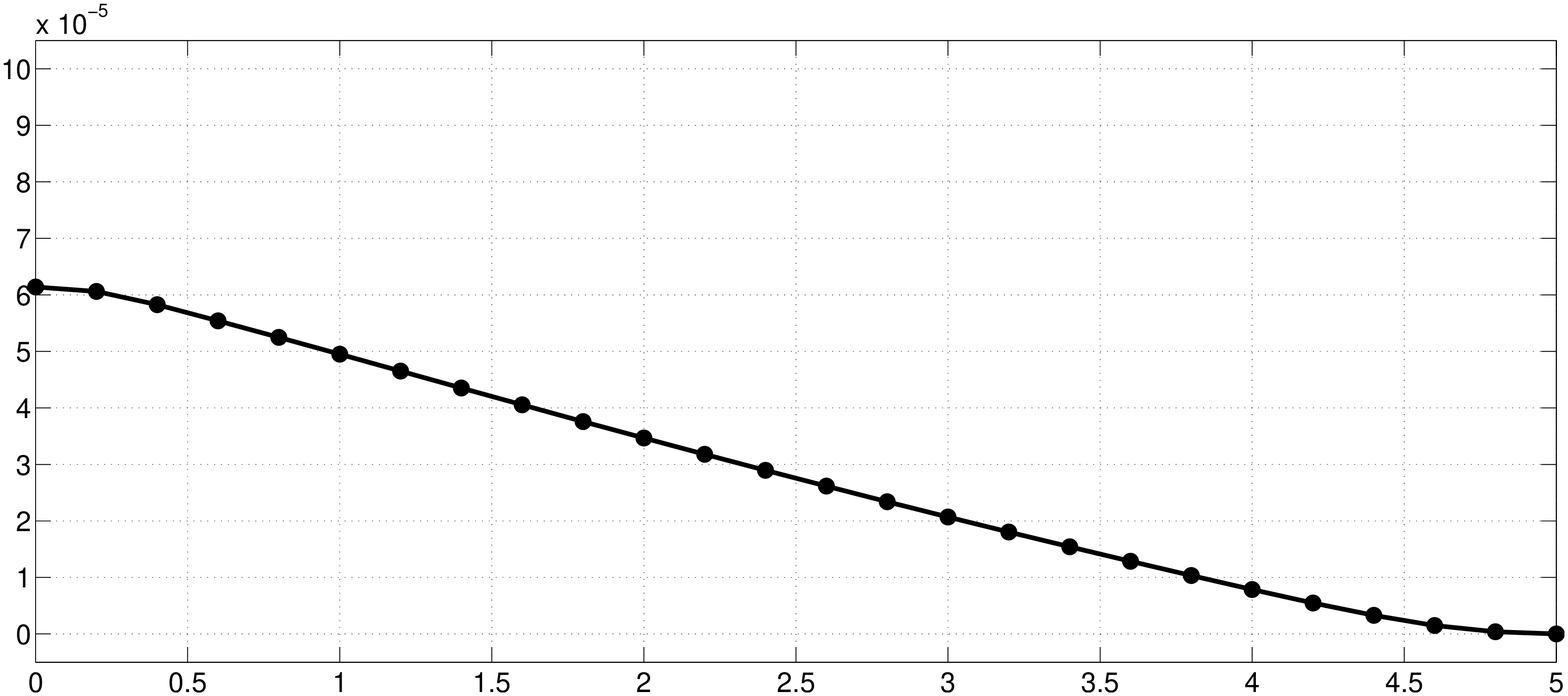,width=0.46\linewidth,clip=} \\
\epsfig{file=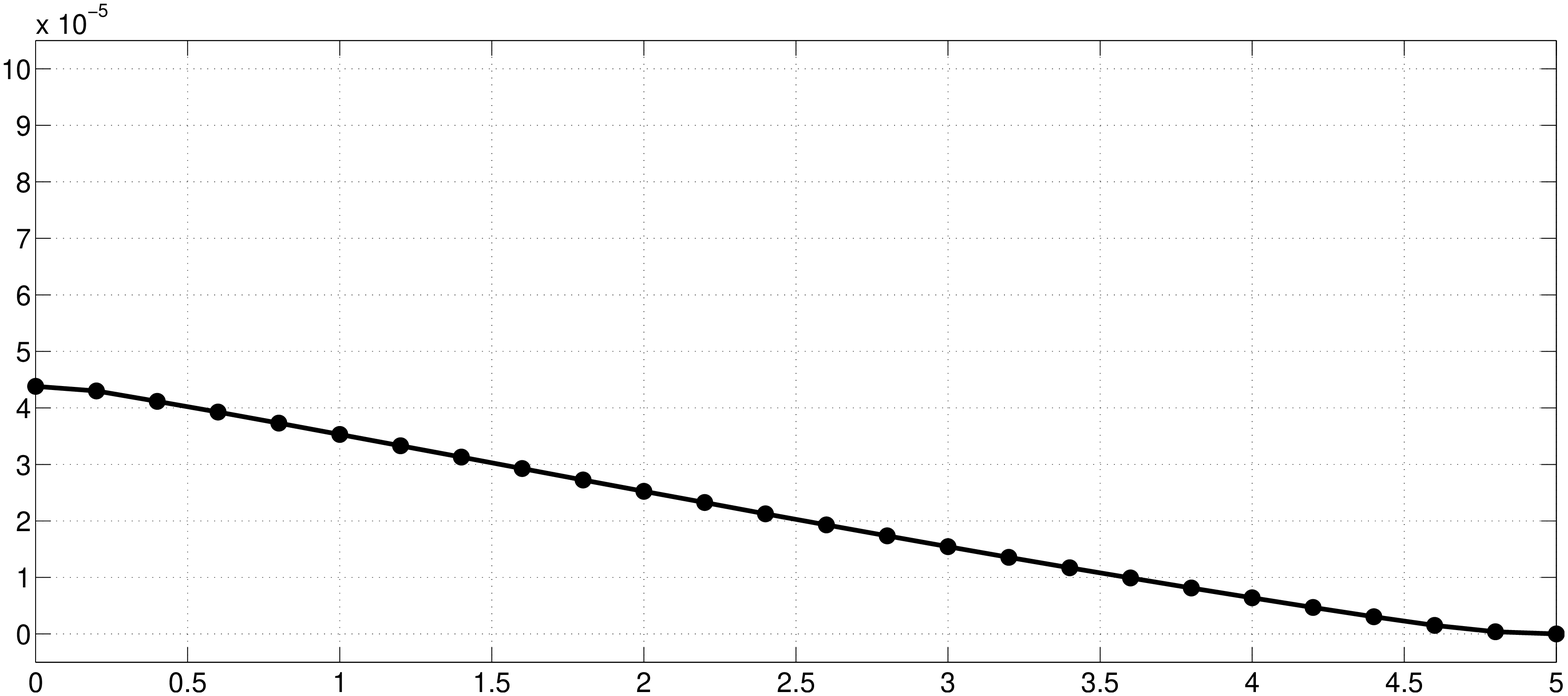,width=0.46\linewidth,clip=} &
\epsfig{file=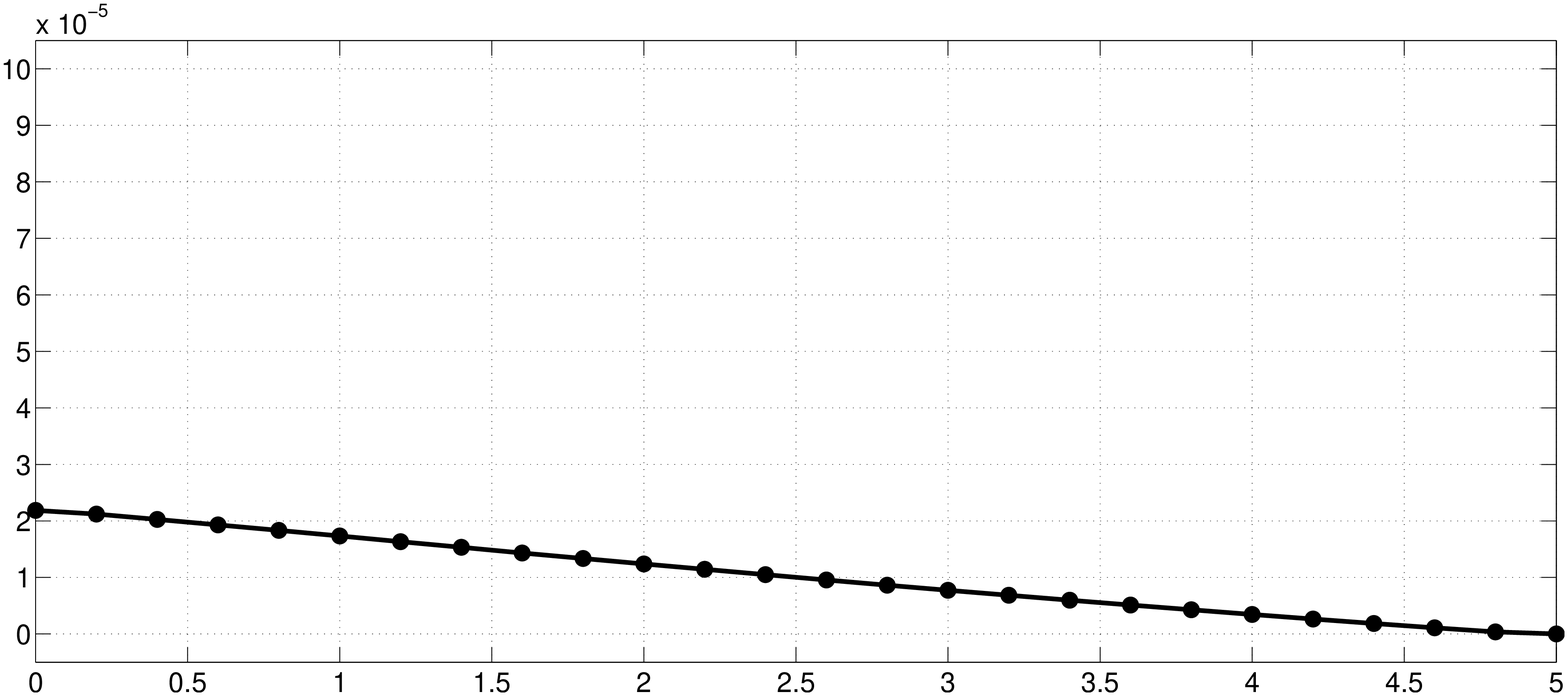,width=0.46\linewidth,clip=} \\
\epsfig{file=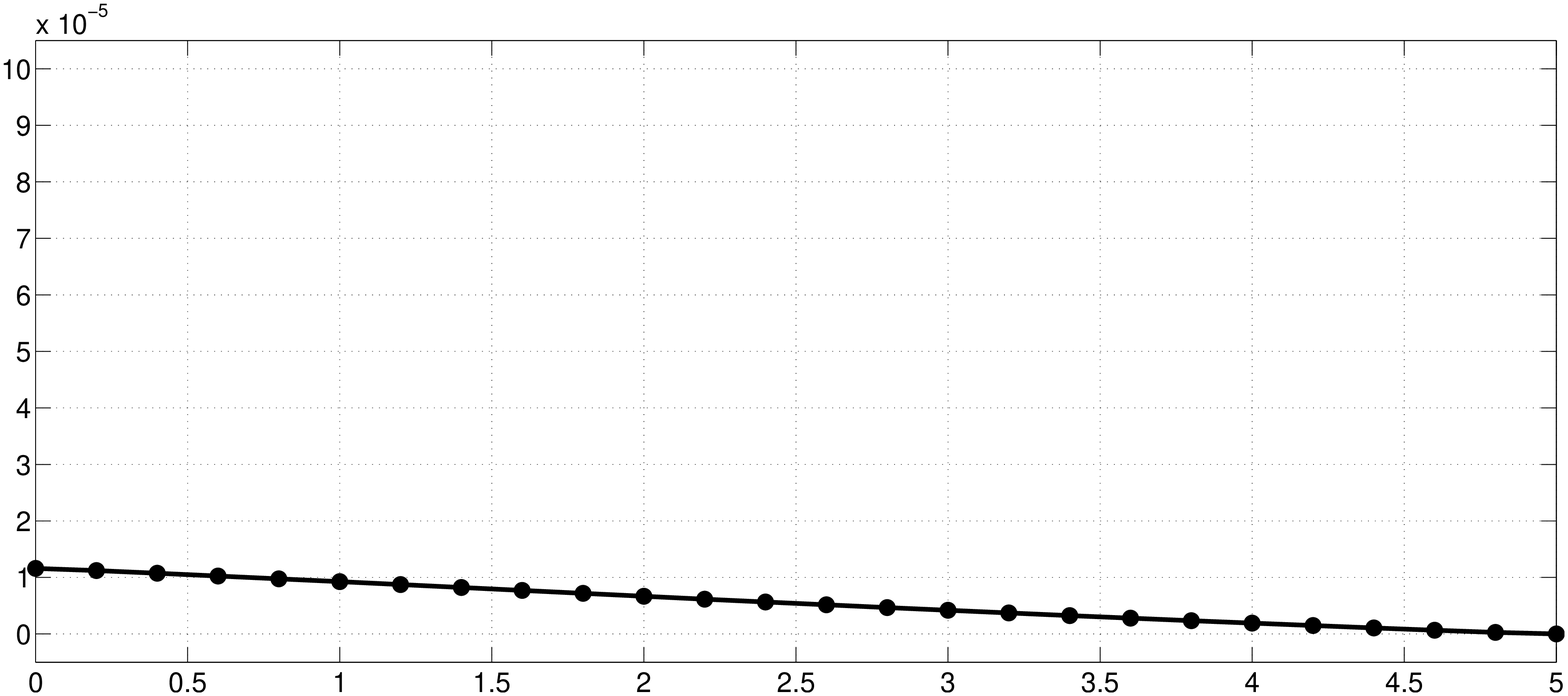,width=0.46\linewidth,clip=} &
\epsfig{file=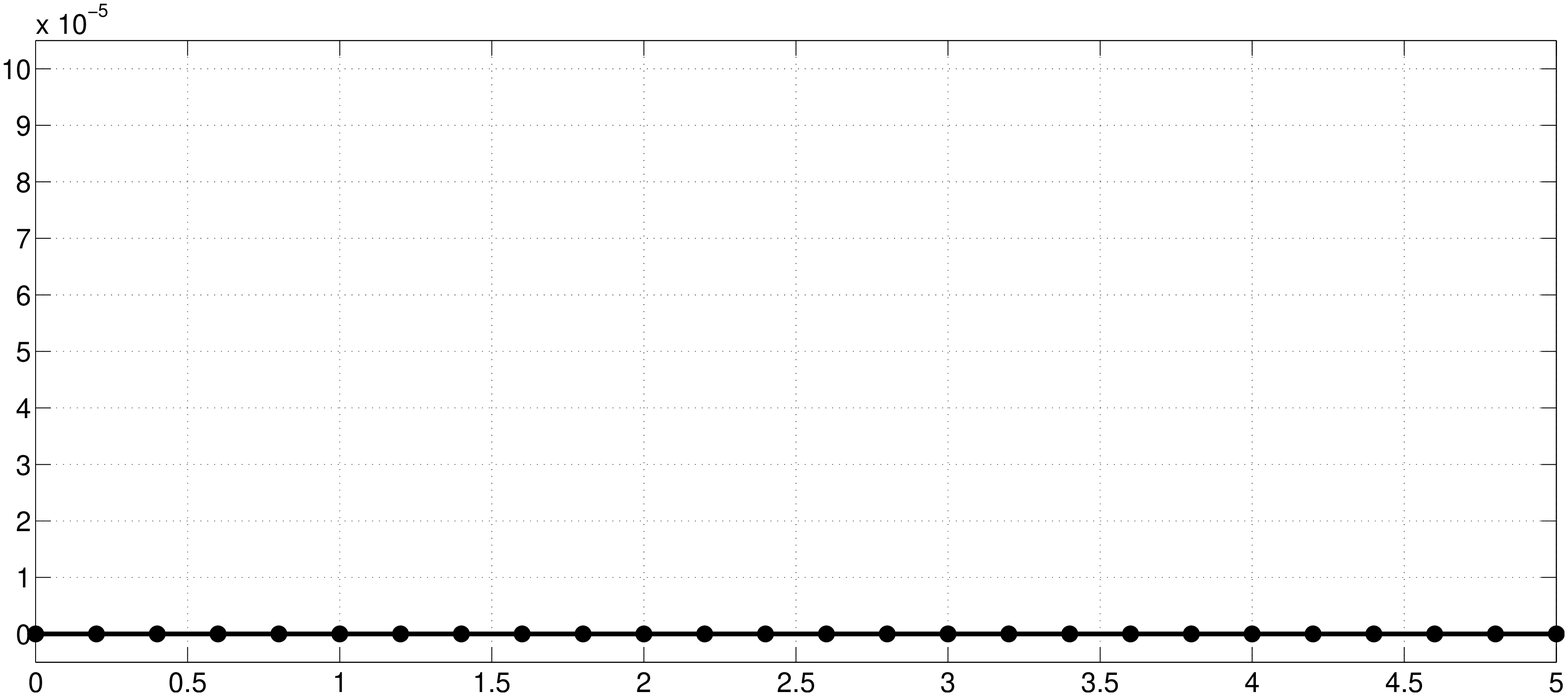,width=0.46\linewidth,clip=}
\end{tabular}
\caption{Forward CVA curves for Case A, B, C, D, E, and F}
\label{fig2}
\end{figure}

\section{Conclusion}

In this paper, we discussed the modeling of counterparty risk in the presence of bilateral margin agreements. We defined an appropriate collateral process which takes various margin agreement parameters into account. The dynamics of the counterparty risk adjustment, CVA, has been found for the bilateral case. This achievement helps us to better understand and monitor the behavior of the bilateral CVA as well as the unilateral CVA and the DVA.

We observed the impact of collateral agreements on counterparty risk adjustments as well as the credit exposures such as the EPE and the ENE. The presence of simultaneous defaults in our model represents the wrong way risk involved in the CDS contracts. We formulate the fair spread value adjustment, which is named as SVA, that indicates the additional spread value to incorporate the counterparty risk into the fair spread value. Moreover, we derive the dynamics of the fair spread and the counterparty risky spread and therefore the spread value adjustment, SVA. Finally, as in \cite{Bielecki2011} and \cite{Assefa2011}, we present our numerical results using a Markovian model of counterparty credit risk.

\bibliographystyle{alpha}

\end{document}